\documentclass[conferemce]{IEEEtran}
\usepackage{graphicx}
\usepackage{amsthm}
\usepackage{epsfig}
\usepackage{latexsym}
\usepackage{amsfonts}
\usepackage{here}
\usepackage{rawfonts}
\usepackage[latin1]{inputenc}
\usepackage[T1]{fontenc}
\usepackage{calc}
\usepackage{capitalgreekitalic}
\usepackage{url}
\usepackage{enumerate}
\usepackage{color}
\usepackage[tbtags]{amsmath}
\usepackage{amssymb}
\usepackage{upref}
\usepackage{epic,eepic}
\usepackage{times}
\usepackage{dsfont}
\usepackage{comment}
\usepackage{cite}
\usepackage{bbm} 
\usepackage{amsmath}
\usepackage{microtype} 
\usepackage{afterpage}
\usepackage{makecell}

\newtheorem{theorem}{\bf Theorem}

\newtheorem{lemma}{\bf Lemma}

\newtheorem{assumption}{\bf Assumption}

\usepackage{dsfont}

\newcounter{step}
\newlength{\totlinewidth}
\newenvironment{algorithm}{%
  \rule{\linewidth}{1pt}
  \begin{list}{}%
    {\usecounter{step}%
      \settowidth{\labelwidth}{\textbf{Step 2:}}%
      \setlength{\leftmargin}{\labelwidth}%
      \setlength{\topsep}{-2pt}%
      \addtolength{\leftmargin}{\labelsep}%
      \addtolength{\leftmargin}{2mm}%
      \setlength{\rightmargin}{2mm}%
      \setlength{\totlinewidth}{\linewidth}%
      \addtolength{\totlinewidth}{\leftmargin}%
      \addtolength{\totlinewidth}{\rightmargin}%
      \setlength{\parsep}{0mm}%
      \raggedright}}%
  {\end{list}%
  \rule{\linewidth}{1pt}}
\newcounter{substep}

  {\end{list}}

\newlength{\aligntop}
\newlength{\alignbot}
 

\IEEEoverridecommandlockouts

\usepackage{algorithm}
\usepackage{algorithmic}

\makeatletter
\newcommand\semihuge{\@setfontsize\semihuge{19.3}{25}}
\makeatother

\makeatletter
\newcommand\semismall{\@setfontsize\semihuge{12.4}{15}}
\makeatother

\usepackage{subfigure}

\begin{document}

\title{\huge Complex-Valued Neural Network based Federated Learning for Multi-user Indoor Positioning Performance Optimization\vspace*{0.2em}}

\author{\large{Hanzhi Yu, Yuchen Liu, \textit{Member IEEE}, and Mingzhe Chen, \textit{Member IEEE} \vspace*{1em}\\ 
}\vspace*{-2em}

\thanks{Hanzhi Yu and Mingzhe Chen are with the Department of Electrical and Computer Engineering and Frost Institute for Data Science and Computing, University of Miami, Coral Gables, FL 33146 USA (Emails: \protect\url{{hanzhiyu, mingzhe.chen}@miami.edu)}.} 
\thanks{Yuchen Liu is with the Department of Computer Science, North Carolina State University, Raleigh, NC 27695 USA (Email: \protect\url{yuchen.liu@ncsu.edu}).}
\thanks{This work was supported by the U.S. National Science Foundation under Grants CNS-2312139 and CNS-2312138. }
}
\maketitle
%
\begin{abstract}
In this article, the use of channel state information (CSI) for indoor positioning is studied. In the considered model, a server equipped with several antennas sends pilot signals to users, while each user uses the received pilot signals to estimate channel states for user positioning. To this end, we formulate the positioning problem as an optimization problem aiming to minimize the gap between the estimated positions and the ground truth positions of users. To solve this problem, we design a complex-valued neural network (CVNN) model based federated learning (FL) algorithm. Compared to standard real-valued centralized machine learning (ML) methods, our proposed algorithm has two main advantages. First, our proposed algorithm can directly process complex-valued CSI data without data transformation. Second, our proposed algorithm is a distributed ML method that does not require users to send their CSI data to the server. Since the output of our proposed algorithm is complex-valued which consists of the real and imaginary parts, we study the use of the CVNN to implement two learning tasks. First, the proposed algorithm directly outputs the estimated positions of a user. Here, the real and imaginary parts of an output neuron represent the 2D coordinates of the user. Second, the proposed method can output two CSI features (i.e., line-of-sight/non-line-of-sight transmission link classification and time of arrival (TOA) prediction) which can be used in traditional positioning algorithms. Simulation results demonstrate that our designed CVNN based FL can reduce the mean positioning error between the estimated position and the actual position by up to 36\%, compared to a RVNN based FL which requires to transform CSI data into real-valued data.
\end{abstract}


\section{Introduction}
Device positioning plays an important role for many emergent applications, such as virtual reality, autonomous vehicles, and shared mobility (e.g., e-scooter rental on Uber) \cite{7762095}. In particular, global navigation satellite system (GNSS) based localization methods particularly global positioning system (GPS) based methods are widely used for these emergent applications. However, GNSS based methods may not be applied for indoor positioning since the signals that are transmitted from satellites to a target user and used for positioning have a higher probability of being blocked by obstacles such as walls, furniture, and human bodies compared to the use of GNSS based methods for outdoor positioning \cite{8451859}. To address this challenge, radio frequency (i.e., WiFi and visible light) based indoor positioning methods is a promising technology due to their ability to capture signal variances in complex indoor environment \cite{8057286,zhu2024survey}. However, the use of radio frequency for indoor positioning still faces several challenges. First, the accuracy of radio frequency based methods depend on line-of-sight (LOS) pilot signal transmission. Non-line-of-sight (NLOS) pilot signal transmission may have high attenuation and signal scattering thus reducing positioning accuracy. Second, radio frequency based positioning may suffer from interference caused by devices that use the same frequency for data transformation \cite{8692423}. To overcome these challenges, one can study the use of fingerprinting-based positioning methods, \cite{9149443, 7438932, 9129126, 8027020, 9535306, 8919897, 9999279, 8761305, 9066152, 9917443, 10214616, 9593115, 10118848, 9838945, 10005038, 9148111} to estimate positions of a user. \\

Recently, a number of existing works \cite{9149443, 7438932, 9129126, 8027020, 9535306, 8919897, 9999279, 8761305} have studied the use of radio frequency and machine learning (ML) tools \cite{9562559} for indoor positioning. In particular, the authors in \cite{9149443} introduced a k-nearest neighbor based positioning method which uses the magnitude component of channel state information (CSI) to estimate the position of a user. In \cite{7438932}, the authors developed a Bolzmann machine based positioning scheme that uses CSI signal amplitudes to estimate the position of a user. The authors in \cite{9129126} designed a convolutional neural network (CNN) based positioning method that uses CSI signals in polar domain to estimate the position of a user. In \cite{8027020}, the authors proposed a CNN based positioning method and considered CSI amplitudes from three antennas as an image. In \cite{9535306}, the authors trained different neural network models by real-valued CSI features which are extracted from the CSI data obtained by different access points, to estimate a probability distribution of the user at given locations as the output. In \cite{8919897}, the authors designed a Siamese neural network based framework for supervised, semisupervised positioning as well as unsupervised channel charting. In \cite{9999279}, the author introduced a method that uses an unsupervised deep autoencoder based model to extract CSI features from CSI data, and uses the extracted features to estimate the position of the user. In \cite{8761305}, the authors presented an outdoor positioning method which first clusters CSI and received signal strength indicator (RSSI) samples into a group using a $K$ nearest neighbors algorithm, and then estimate the position of the user using a deep neural network trained by the samples in the same group. However, most of these existing works \cite{9149443, 7438932, 9129126, 8027020, 9535306, 8919897, 9999279, 8761305} need to transform complex-valued CSI data into real-valued data so as to feed into real-valued ML models by: 1) separating the real and imaginary part, 2) using the power of the real and imaginary parts, 3) converting the complex-valued CSI into polar domain values. These transformation methods may lose the features in original complex-valued CSI data thus decreasing accuracy of the positioning algorithms. Moreover, all these works \cite{9149443, 7438932, 9129126, 8027020, 9535306, 8919897, 9999279, 8761305} require users to send their collected CSI data to a server, which may not be practical since not all users are willing to share their personal data due to privacy and security issues. \\

To avoid users transmitting original CSI data for user positioning, a number of existing works \cite{9066152, 9917443, 10214616, 9593115, 10118848, 9838945, 10005038, 9148111} have studied the use of federated learning (FL) for user localization. In particular, the authors in \cite{9066152} designed an FL based fingerprinting method which uses received signal strength indicator (RSSI) data to train a neural network model that consists of an autoencoder and a deep neural network for user positioning. In \cite{9917443}, the authors introduced a convolutional neural network (CNN) based FL which uses received signal strength (RSS) data to solve the problem of building-floor classification and latitude-longitude regression in indoor localization. In \cite{10214616}, the authors designed a personalized FL which trains a reinforcement learning model at each user device using RSS data for indoor positioning. In \cite{9593115}, the authors designed a federated attentive message passing method which trains a personalized local model for each user via its non-independent and identically distributed (non-IID) RSS data to estimate its position. In \cite{10118848}, the authors designed a hierarchical multilayer perceptron (MLP) based FL positioning algorithm which uses the RSSI as input and outputs the estimation of the building, the floor, and the 2D location where the user is located. In \cite{9838945}, the authors proposed a FL based CSI fingerprinting method, which trains two local CNNs of each access point to extract features separately from the amplitude and the phase difference of the CSI, and uses a global fully connected estimator to estimate the location of the user. In \cite{10005038}, the authors introduced a FL framework for indoor positioning, which uses CSI amplitude as input, and the posterior probability of the user at each reference position as the output. In \cite{9148111}, the authors designed a MLP based FL algorithm which uses RSS fingerprints as input to estimate the coordinate of the user. However, all methods in \cite{9066152, 9917443, 10214616, 9593115, 10118848, 9838945, 10005038, 9148111} are real-valued ML algorithms which need to be trained by real-valued data. Hence, the methods in \cite{9838945, 10005038} may not be able to extract all features from the original complex-valued CSI data. To this end, it is necessary to design a complex-valued neural network (CVNN) based FL which process complex-valued CSI without data transformation. However, designing CVNN based FL faces several unique challenges. First, when training a CVNN, we need to calculate the complex-valued gradients and update complex-valued weights. Hence, it is necessary to design novel methods to calculate gradients and update complex-valued weights \cite{9413814}. Second, real-valued activation functions designed for real-valued neural networks (RVNNs) cannot be directly used for CVNNs due to the conflict between the boundedness and the analyticity of complex functions. Hence, a proper activation function in the complex plane must be designed for CVNNs \cite{benvenuto1992complex}. Third, since the output of a CVNN is complex-valued, one can output two values (i.e., real and imaginary values). Therefore, it is necessary to design the output of the CVNN so as to obtain better user positioning performance. \\

The main contribution of this work is to design a novel indoor positioning framework that uses complex-valued CSI data to estimate the position of users. Our key contributions include: 
\begin{itemize}
    \item We consider an indoor positioning system where the server equipped with several antennas sends pilot signals to users. Each user receives pilot signals to estimate the channel states. The estimated CSI is used to predict the position of the user. The goal of our designed system is to train a ML model which uses a set of CSI data to predict the position of users. To this end, we formulate an optimization problem aiming to minimize the mean square error between predicted positions and ground truth positions of users. 
    \item To solve the formulated problem, we proposed a novel CVNN model based FL algorithm. Compared to traditional real-valued centralized ML methods  \cite{9149443, 7438932, 9129126, 8027020, 9535306, 8919897, 9999279, 8761305}, our designed method has two key advantages. First, our designed method is a distributed ML method which enables the users to train their CVNN models without local CSI data and position information sharing. Second, an FL model at each device is a CVNN which can process original complex-valued CSI data without any data transformation. Hence, compared to RVNN based positioning methods which must transform complex-valued CSI data into real-valued data, our proposed algorithm can extract more CSI features thus improving positioning accuracy. Furthermore, to reduce the communication overhead of FL model parameter transmission, we design a novel FL parameter model transmission scheme which enables each device to transmit only real or imaginary parts of the FL models to the server. 
    \item We propose to use our designed CVNN model based FL algorithm for two use cases: 1) the output of our designed algorithm is the estimated position of the user. Here, the real and imaginary parts of the output neuron separately represents the x-coordinate and y-coordinate of the user; 2) the output of our designed algorithm extracts two CSI features such as time of arrival (TOA) and LOS/NLOS transmission link classification that can be used for traditional positioning algorithms. Here, the real and imaginary parts of the output neuron can represent two CSI features. This is a major differences between CVNN and RVNN since one real-valued neuron in a RVNN model can represent only one CSI feature. 
    \item We derive a closed-form expression for the expected convergence rate of our designed CVNN based FL algorithm and build an explicit relationship between the probability of each user transmitting the real or imaginary part of the model parameters to the server and the performance of the FL algorithm. \\
\end{itemize}

Simulation results show that our proposed CVNN based FL positioning method can reduce the mean positioning error between the estimated position and the actual position by up to 36\% compared to a RVNN model based positioning algorithm. \\

The remainder of this paper is organized as follows. The system model and problem formulation are introduced in Section \uppercase\expandafter{\romannumeral2}. The design of the CVNN based FL model will be introduced in Section \uppercase\expandafter{\romannumeral3}. The expected convergence rate of our designed CVNN based FL is studied in Section \uppercase\expandafter{\romannumeral4}, simulation results are presented and discussed in Section \uppercase\expandafter{\romannumeral5}. Finally, conclusions are drawn in Section \uppercase\expandafter{\romannumeral6}. 

\section{System Model and Problem Formulation}\label{se:system}
%


%
%
Consider a positioning system in which a server uses CSI to estimate positions of a set $\mathcal{U}$ of $U$ users, as shown in Fig. \ref{fig:environment}. The server equipped with $C$ antennas sends pilot signals to users over $L$ subcarriers. Each user receives pilot signals to estimate the channel states. By using the estimated CSI, the position of the user can be predicted. In the considered scenario, several obstacles exist and hence may block the transmission of pilot signals between the users and the server. Hence, the transmission link between a user and the server may be NLOS. Next, we first introduce the process of CSI collection. Then, we introduce our considered positioning problem. Table \ref{tab:notation} provides a summary of the notations used hereinafter. 
\begin{figure}[!t]
 \begin{center}
    \includegraphics[width=7cm]{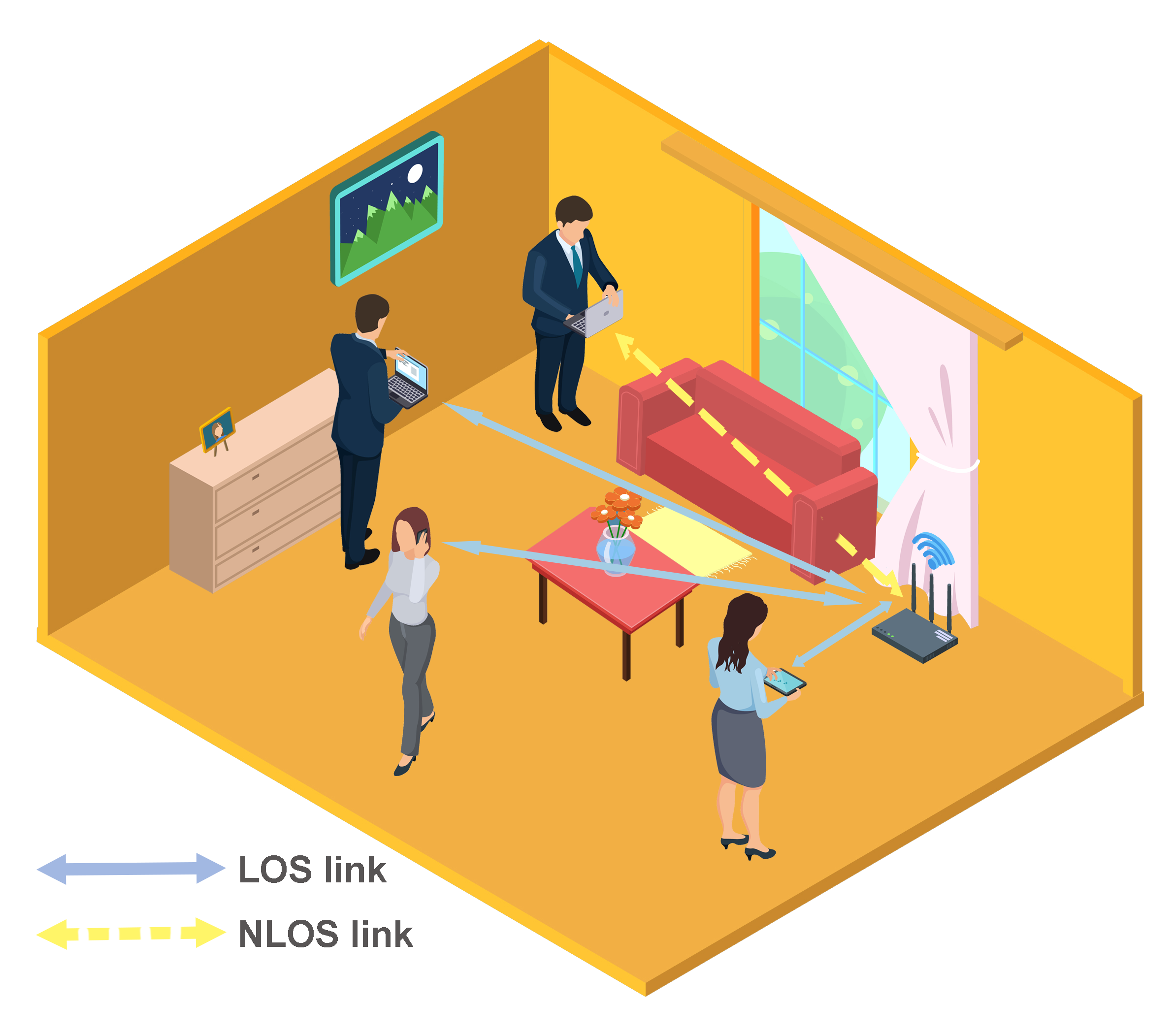}
    \caption{The considered indoor positioning system. }
    \label{fig:environment} 
    \vspace{-.5cm}
  \end{center}
\end{figure}

\begin{table*}
\caption{List of Notations}
\label{tab:notation}
\centering
\begin{tabular}{|c|c|c|c|}
\hline
\textbf{Notation} & \textbf{Description} & \textbf{Notation} & \textbf{Description}\\
\hline

$C$ & Number of antennas equipped on the server & $U$ & Number of users \\
\hline
$\boldsymbol{w}_u$ & ML model parameters of user $u$ & $\mathcal{U}$ & Set of users \\
\hline
$\boldsymbol{Y}^{\textrm{I}}, \boldsymbol{Y}^{\textrm{II}}$ & Output of convolutional layer \textrm{I} and \textrm{II} & $L$ & Number of subcarriers\\
\hline
$\boldsymbol{W}^{\textrm{I}}, \boldsymbol{W}^{\textrm{II}}$ & Convolutional kernels of convolutional layers \textrm{I} and \textrm{II} & $\boldsymbol{H}$ & The CSI matrix \\
\hline
$\boldsymbol{b}^{\textrm{I}}, \boldsymbol{b}^{\textrm{II}}$ & Bias vectors of convolutional layers \textrm{I} and \textrm{II} & $\boldsymbol{h}_j$ & Column or row $j$ of $\boldsymbol{H}$ \\
\hline
$S_1 \times S_1$ & Kernel size of convolutional layer \textrm{I} & $\hat{\boldsymbol{H}}$ & Normalized $\boldsymbol{H}$ \\
\hline
$S_2 \times S_2$ & Kernel size of convolutional layer \textrm{II} & $\mathcal{D}_u$ & Local dataset of user $u$ \\
\hline
$O^{\textrm{I}}$ & Number of the output channel of convolutional layer \textrm{I} & $\left( \boldsymbol{H}_{u,k}, \boldsymbol{p}_{u,k} \right)$ & Data sample $k$ in $\mathcal{D}_u$ \\
\hline
$O^{\textrm{II}}$ & Number of the output channel of convolutional layer \textrm{II} & $\overline{\boldsymbol{Y}}^{\textrm{I}}, \overline{\boldsymbol{Y}}^{\textrm{II}}$ & Outputs of pooling layers \\
\hline
$S^{\textrm{I}}, S^{\textrm{II}}$ & Stride size of pooling layer \textrm{I} and \textrm{II} & $\boldsymbol{y}^{\textrm{III}}$ & Output of the flatten layer \\
\hline
$P^{\textrm{I}}, P^{\textrm{II}}$ & Size of the pooling window of pooling layer \textrm{I} and \textrm{II} & $\boldsymbol{y}', \boldsymbol{y}''$ & Outputs of fully connected layers \\
\hline
$\boldsymbol{W}', \boldsymbol{W}''$ & Weight matrices of fully connected layers & $\hat{y}$ & Output of the CVNN model \\
\hline
$\boldsymbol{b}', \boldsymbol{b}''$ & Bias vectors of fully connected layers & $\boldsymbol{g}$ & Global FL model \\
\hline
$N^{\textrm{I}}, N^{\textrm{II}}$ & Number of neurons of fully connected layers & $T$ & Number of FL training iterations \\
\hline
$\boldsymbol{w}$ & Weight matrix of the CVNN output layer & $b$ & Bias of the CVNN output layer \\
\hline
$\overline{\boldsymbol{W}}$ & All parameters of the CVNN model & $\overline{\boldsymbol{W}}_{u}^{t}$ & Local FL model of user $u$ \\
\hline
$r_u^t$ & \makecell[c]{A variable indicates whether user $u$ transmits \\ the real part of $\overline{\boldsymbol{W}}_u^t$ to the server} & $T$ & Number of FL training iterations \\
\hline
$m_u^t$ & \makecell[c]{A variable indicates whether user $u$ transmits \\ the imaginary part of $\overline{\boldsymbol{W}}_u^t$ to the server}& $\mathcal{B}_{u}^{t}$ & \makecell[c]{The training batch of \\ user $u$ at iteration $t$} \\
\hline
\end{tabular}
\end{table*}

\subsection{CSI Data Collection}
We assume that the server and the user communicate on a narrowband flat-fading channel. Let $\boldsymbol{x} \in \mathbb{C}^{L \times 1}$ be the symbol vector transmitted from the server to a user. Then, the signal received by a user is  
\begin{equation}\label{eq:Y}
    \boldsymbol{y} = \boldsymbol{H} \boldsymbol{P} \boldsymbol{x} + \boldsymbol{n},  
\end{equation} 
where $\boldsymbol{P} = \left[ \boldsymbol{p}_1,...,\boldsymbol{p}_L \right]$ is the beamforming matrix with $\boldsymbol{p}_l \in \mathbb{C}^{C \times 1}$ being the beamforming vector at subcarrier $l$, $\boldsymbol{n} \in \mathbb{C}^{L \times 1}$ is the additive white Gaussian noise, and $\boldsymbol{H} \in \mathbb{C}^{L \times C}$ is the CSI matrix of the transmission link between the server and the user over all subcarriers. The CSI matrix $\boldsymbol{H}$ received by a user over $L$ subcarriers is 
\begin{equation}\label{eq:Hf}
    \boldsymbol{H} = \left[ \boldsymbol{h}_1, \boldsymbol{h}_2, ..., \boldsymbol{h}_L \right]^T, 
\end{equation}
where $\boldsymbol{h}_l \in \mathbb{C}^{C \times 1}$ is the channel vector at subcarrier $l$ over $C$ antennas \cite{9076084}. 

\subsection{Problem Formulation}
Given the defined model, next, we introduce our positioning problem. Our goal is to design a ML algorithm which uses the collected CSI to estimate the position of a set of users. We assume that user $u$ has a local dataset $\mathcal{D}_u$ that consists of $|\mathcal{D}_u|$ data samples. Each data sample $k$ of user $u$ consists of CSI matrix $\boldsymbol{H}_{u,k}$ and the user's position $\boldsymbol{p}_{u,k} = \left[ a_{u,k}, b_{u,k} \right]$, where $a_{u,k}$, $b_{u,k}$ are the coordinates of user $u$. Let $f \left( \boldsymbol{w}_u, \boldsymbol{H}_{u,k} \right)$ be the position estimated by the ML algorithm, where $\boldsymbol{w}_u$ is the ML model parameters of user $u$. Then, the positioning problem can be formulated as an optimization problem whose goal is to minimize the gap between actual positions and estimated positions of $U$ users, which can be expressed as 
\begin{equation}\label{eq:problem}
    \mathop{\min}_{\boldsymbol{w}_1,...,\boldsymbol{w}_U} \frac{1}{\sum_{u=1}^U |\mathcal{D}_u|} \sum_{u=1}^{U} \sum_{k=1}^{|\mathcal{D}_u|} {\lVert f(\boldsymbol{w}_u,\,\boldsymbol{H}_{u,k}) - \boldsymbol{p}_{u,k} \rVert}_{2}^{2}. 
\end{equation}
To solve problem (\ref{eq:problem}), several methods such as in \cite{9149443, 7438932, 9129126, 8027020, 9535306, 8919897, 9999279, 8761305, 9593115, 10118848, 9838945, 10005038, 9148111, 9066152, 9917443, 10214616} have already been proposed. However, these methods have two key limitations. First, these methods are designed based on RVNNs and hence these methods need to first transform the complex-valued CSI data to real-valued CSI data, by: 1) separating the real part and the imaginary part of the complex-valued data, 2) using absolute values of the complex-valued data, or 3) converting the complex-valued data into polar domain values. However, the transformation of complex-valued CSI data to real-valued CSI data may lose some features of the original complex-valued CSI data thus decreasing the prediction accuracy of ML models. To solve this problem, we proposed a novel CVNN, which can directly use the original complex-valued CSI data as the input of the model without the transformation of data from complex-valued to real-valued. Second, most of current positioning methods \cite{9149443, 7438932, 9129126, 8027020, 9535306, 8919897, 9999279, 8761305} require the users to transmit their CSI data and the corresponding ground truth positions to the server for ML model training, which may not be practical since most of the users may not want to share their position information with the server due to privacy and security concerns. To address this issue, we propose to combine the designed CVNN model with FL which enables the server and a set of users to learn a common CVNN model cooperatively without requiring users to transmit their collected CSI data and positions to the server. 

\section{Proposed Complex-valued Neural Network Based FL}
In this section, we introduce the proposed CVNN based FL algorithm for solving problem (\ref{eq:problem}). 
Compared to traditional real-valued centralized ML methods \cite{9149443, 7438932, 9129126, 8027020, 9535306, 8919897, 9999279, 8761305}, our proposed method has two key advantages. First, our proposed method can directly process complex-valued CSI data without any data transformation from complex values to real values thus extracting more CSI features from CSI data and improving position prediction accuracy. Second, our designed positioning method is a distributed method which does not require users to transmit CSI information to the server during the model training process. Next, we first introduce the components of the proposed local FL model of each user. Then, we explain the training process of the designed algorithm. 

\subsection{Components of the Local FL Model}
Here, we introduce the components of the designed CVNN based local FL model, which consists of the following components: a) input layer, b) convolutional layer \textrm{I}, c) pooling layer \textrm{I}, d) convolutional layer \textrm{II}, e) pooling layer \textrm{II}, f) flatten layer, g) fully connected layer \textrm{I}, h) fully connected layer \textrm{II}, and i) output layer. These components are specified as follows: 
\begin{itemize}
    \item \textbf{Input layer}: To estimate the position of the user, the input of the designed CVNN based FL model at each device is the complex-valued CSI matrix $\boldsymbol{H}$. For each CSI sample, the normalization of complex-valued CSI matrix $\boldsymbol{H}$ is 
    \begin{equation}\label{eq:norm}
        \hat{\boldsymbol{h}}_j = \frac{\mathfrak{R}\left(\boldsymbol{h}_j\right)}{\max\left(\left|\boldsymbol{h}_j\right|\right)} + i \frac{\mathfrak{I}\left(\boldsymbol{h}_j\right)}{\max\left(\left|\boldsymbol{h}_j\right|\right)}, 
    \end{equation}
    where $\boldsymbol{h}_j$ is column or row $j$ of $\boldsymbol{H}$, $\hat{\boldsymbol{h}}_j$ is the corresponding column or row $j$ of the normalized CSI matrix $\hat{\boldsymbol{H}}$, $\mathfrak{R}\left( \boldsymbol{h}_j \right)$ is the real part of $\boldsymbol{h}_j$, and $\mathfrak{I}\left( \boldsymbol{h}_j \right)$ is the imaginary part of $\boldsymbol{h}_j$. If $\boldsymbol{h}_j$ is column $j$ of $\boldsymbol{H}$, we use (\ref{eq:norm}) to normalize the CSI matrix over each antenna. Otherwise, the normalization of the CSI sample $\boldsymbol{H}$ is over each CSI feature. The use of antenna normalization or feature normalization depends on the specific dataset.  The normalization method in (\ref{eq:norm}) is different from common normalization methods used in RVNNs since they cannot process complex numbers. 
    \item \textbf{Convolutional layer \textrm{I}}: We first use a 2D convolutional layer to extract CSI features. The relationship between the input $\hat{\boldsymbol{H}}$ and the output $\boldsymbol{Y}^{\textrm{I}} \in \mathbb{C}^{O^{\textrm{I}} \times C^{\textrm{I}} \times L^{\textrm{I}}}$ of this layer is 
    \begin{equation}\label{eq:conv1d}
        \boldsymbol{Y}_i^{\textrm{I}} = \phi_{1}\left( b_i^{\textrm{I}} + \boldsymbol{W}_{i,0}^{\textrm{I}} * \hat{\boldsymbol{H}} \right), 
    \end{equation}
    where $O^{\textrm{I}}$ is the number of the output channel, $\boldsymbol{Y}_{i}^{\textrm{I}} \in  \mathbb{C}^{C^{\textrm{I}} \times L^{\textrm{I}}}$ is matrix $i$ of the 3D matrix $\boldsymbol{Y}^{\textrm{I}}$ with $C^{\textrm{I}}$ being the height of each matrix $i$ and $L^{\textrm{I}}$ being the width, $\boldsymbol{W}_{i,0}^{\textrm{I}} \in \mathbb{C}^{S_1 \times S_1}$ is a weight parameters matrix of the convolutional kernel $\boldsymbol{W}^{\textrm{I}} \in \mathbb{C}^{O^{\textrm{I}} \times 1 \times S_1 \times S_1}$ with $S_1 \times S_1$ being the size of the convolutional kernel, $b_i^{\textrm{I}}$ is component $i$ of the bias vector $\boldsymbol{b}^{\textrm{I}} \in \mathbb{C}^{O^{\textrm{I}} \times 1}$, and $\phi_{1} \left( z \right)$ is a complex-valued activation function with respect to a complex number $z$. The activation function $\phi_{1} \left( z \right)$ is a variation of the ReLU function, which is defined as: 
    \begin{equation} \label{eq:phi1}
        \phi_{1}\left(z\right) = \max\left(0,\mathfrak{R}\left(z\right)\right) + i \max\left(0,\mathfrak{I}\left(z\right)\right). 
    \end{equation}
    From (\ref{eq:phi1}), we see that $\phi_{1}\left(z\right)$ is a complex-valued ReLU activation function that separately processes the real and imaginary part of complex number $z$. Here, we can also consider other types of complex-valued ReLU activation functions such as modReLU which is defined as 
    \begin{equation}\label{eq:modReLU}
        \psi\left(z\right) = \left\{
            \begin{aligned}
                &\left(\left|z\right|+q\right) \frac{z}{\left|z\right|} &&\text{if} \ \left|z\right|+q \geq 0 \text{,}  \\
                &0 &&\text{otherwise. }\\
            \end{aligned}
            \right.
    \end{equation}
    where $\left| z \right|$ is the absolute value (or modulus or magnitude) of the complex number $z$, and $q \in \mathbb{R}$ is a learnable parameter. 
    \item \textbf{Pooling layer \textrm{I}}: Pooling layers are used to reduce the size of the output of the previous convolutional layer. Here, we use an average pooling method to process the output of convolutional layer $\textrm{I}$. The size of the pooling window is $P^{\textrm{I}}$, and the size of the stride is $S^{\textrm{I}}$.
    \item \textbf{Convolutional layer \textrm{II}}: The input of convolutional layer \textrm{II} is $\overline{\boldsymbol{Y}}^{\textrm{I}} \in \mathbb{C}^{O^\textrm{I} \times C^\textrm{I} \times \overline{L}^\textrm{I}}$ which is the output of the pooling layer \textrm{I}. The output $\boldsymbol{Y}^{\textrm{II}} \in \mathbb{C}^{O^{\textrm{II}} \times C^{\textrm{II}} \times L^{\textrm{II}}}$ of this layer can be obtained via (\ref{eq:conv1d}), with $O^\textrm{II}$ being the number of the output channel. The parameters of this layer include the convolutional kernel $\boldsymbol{W}^{\textrm{II}} \in \mathbb{C}^{O^{\textrm{II}} \times O^{\textrm{I}} \times S_2 \times S_2}$ with $S_2 \times S_2$ being the size of a convolutional kernel, and the bias vector $\boldsymbol{b}^{\textrm{II}}$. 
    \item \textbf{Pooling layer \textrm{II}}: The input of pooling layer \textrm{II} is $\boldsymbol{Y}^{\textrm{II}}$ which is the output of convolutional layer \textrm{II}. The output of this layer is $\overline{\boldsymbol{Y}}^{\textrm{II}} \in \mathbb{C}^{O^{\textrm{II}} \times C^{\textrm{II}} \times \overline{L}^{\textrm{II}}} $. The size of the pooling window is $P^{\textrm{II}}$, and the size of the stride is $S^{\textrm{II}}$. 
    \item \textbf{Flatten layer}: The flatten layer is used to convert the output $\overline{\boldsymbol{Y}}^{\textrm{II}}$ of the pooling layer \textrm{II} to a row vector $\boldsymbol{y}^{\textrm{III}} \in \mathbb{C}^{1 \times O^{\textrm{II}}C^{\textrm{II}}\overline{L}^{\textrm{II}}}$.  
    \item \textbf{Fully connected layer \textrm{I}}: Fully connected layers are used to learn the relationships among the features extracted by convolutional layers. Given input $\boldsymbol{y}^{\textrm{III}}$, the output is   
    \begin{equation}\label{eq:fc}
        \boldsymbol{y}' = \phi_{1}\left( \boldsymbol{y}^{\textrm{III}} \boldsymbol{W}' + \boldsymbol{b}'\right), 
    \end{equation}
    where $\boldsymbol{W}' \in \mathbb{C}^{O^{\textrm{II}}C^{\textrm{II}}\overline{L}^{\textrm{II}} \times N^{\textrm{I}}}$ is the weight matrix, with $N^{\textrm{I}}$ being the number of neurons in fully connected layer \textrm{I}, $\boldsymbol{y}' \in \mathbb{C}^{1 \times N^{\textrm{I}}}$ is the output vector, and $\boldsymbol{b}' \in \mathbb{C}^{1 \times N^{\textrm{I}}}$ is the bias vector. 
    \item \textbf{Fully connected layer \textrm{II}}: The input of fully connected layer \textrm{II} is $\boldsymbol{y}'$ which is the output of fully connected layer~\textrm{I}. We assume that the number of neurons in this layer is $N^{\textrm{II}}$, and the parameters of this layer are $\boldsymbol{W}''$ and $\boldsymbol{b}''$. Then the relationship between $\boldsymbol{y}'$ and the output of this layer $\boldsymbol{y}''$ can be expressed using (\ref{eq:fc}). 
    \item \textbf{Output layer}: The output of our designed model is 
    \begin{equation}\label{eq:out}
        \hat{y} = \boldsymbol{y}'' \boldsymbol{w} + b, 
    \end{equation}
    where $\boldsymbol{w} \in \mathbb{C}^{N^{\textrm{II}} \times 1}$ is the weight vector, and $b \in \mathbb{C}$ is a bias parameter. The output of our designed model $\hat{y}$ is a complex number which consists of the real part and the imaginary part. Thus, in our proposed positioning system, we consider using different parts of the complex-valued output to work on different learning tasks. To introduce the use of the output $\hat{y}$ in our proposed positioning method, we first rewrite the output $\hat{y}$ as 
    \begin{equation}\label{eq:out_re}
        \hat{y} = \hat{a} + i\hat{b}, 
    \end{equation}
    where $\hat{a} \in \mathbb{R}$ is the real part of $\hat{y}$, and $\hat{b} \in \mathbb{R}$ is the imaginary part of $\hat{y}$. Given (\ref{eq:out_re}), we introduce two use cases of our proposed positioning method: 
    \begin{enumerate}[I.]
        \item The designed model can directly output the coordinates of the estimated position of the user. Therefore, output $\hat{y}$ is a estimated position of the user. To this end, $\hat{a}, \hat{b}$ are the coordinates of the estimated user position. 
        \item The designed algorithm can be used to extract CSI features. These CSI features can be used in traditional positioning algorithms, such as a TOA positioning method \cite{1703954}. Here, $\hat{a} \in \{0,1\}$ is used to identify whether transmission link is LOS. In particular, $\hat{a} = 1$ represents that the transmission link is LOS while $\hat{a} = 0$ represents that the transmission link is NLOS. $\hat{b}$ is the estimated TOA of the signal. In this use case, $\hat{a}$ and $\hat{b}$ are used in different learning tasks. Therefore, one can use the designed algorithm to perform two learning tasks. This is one of the key advantages of our designed algorithm since traditional RVNN based methods can only perform one learning task. 
    \end{enumerate}   
\end{itemize}
\begin{figure}[!t]
  \begin{center}
    \includegraphics[width=8cm]{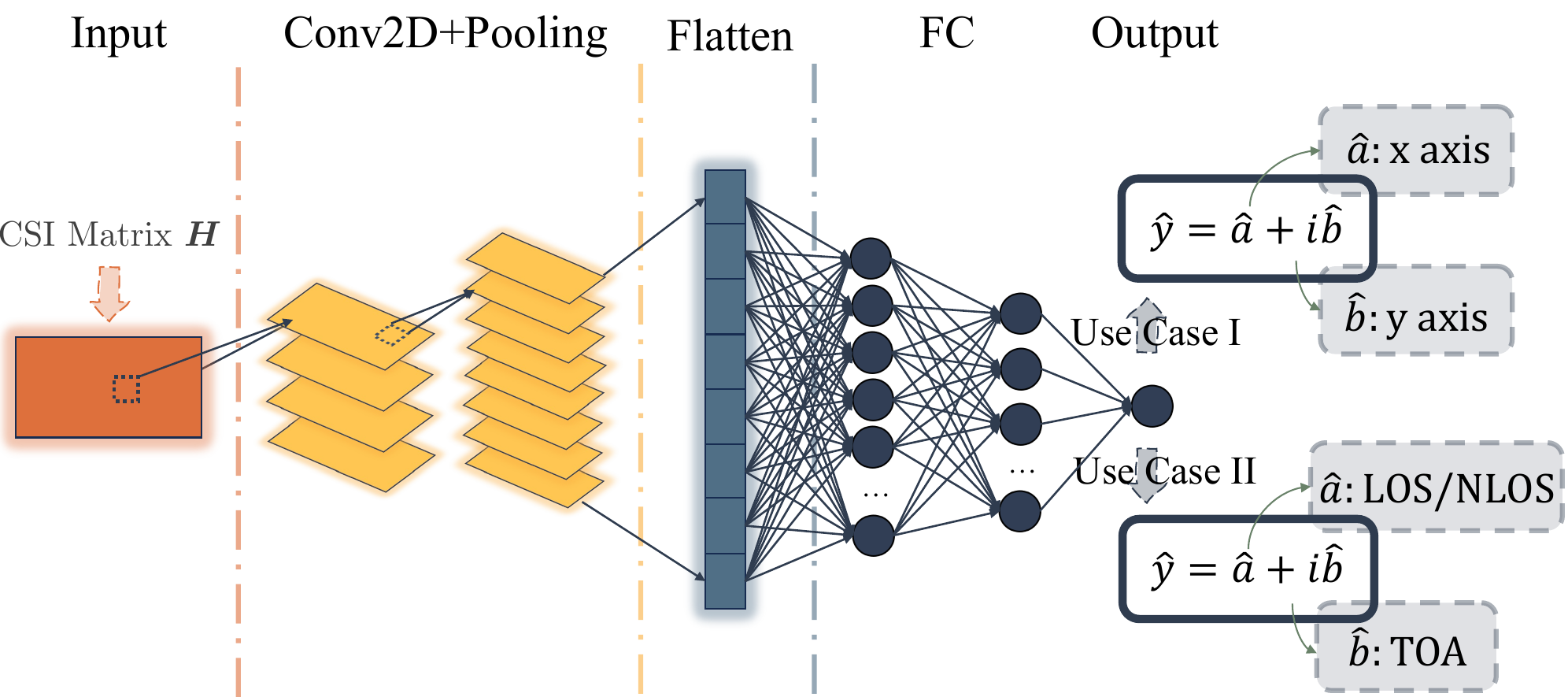}
    \caption{\label{fig:network} The CVNN model structure of use case \textrm{I} and use case \textrm{II}. }
  \end{center}
  \vspace{-0.5cm}
\end{figure}

\subsection{Training Procedure of Federated Learning Algorithm}
Given the local FL model of each device in Section \uppercase\expandafter{\romannumeral3}-A, we next introduce the method of training our designed FL algorithm. First, we introduce the local loss functions used to evaluate the performance of the local FL models over two use cases: \textrm{I}. user position estimation, \textrm{II}. LOS/NLOS transmission link classification and signal TOA estimation. Then, we explain the training process of our designed FL algorithm. \\

\subsubsection{Loss Function for Use Case \textrm{I}} In use case \textrm{I}, the real and the imaginary part of the output of our designed model are coordinates of the estimated user position. Here, we can use one loss function to measure the training loss of the real part and the imaginary part. Since user $u$ has $|\mathcal{D}_u|$ data samples, we assume that the output of the local FL model of user $u$ is $\hat{\boldsymbol{y}} \in \mathbb{C}^{|\mathcal{D}_u| \times 1}$. Then, the total loss function of the local FL model of each user $u$ for case \textrm{I} is given by
\begin{equation}\label{eq:loss1}
    J \left( \overline{\boldsymbol{W}}, \mathcal{D}_u, \boldsymbol{a}, \boldsymbol{b} \right) = \alpha \mathcal{L}_1 \left( \hat{\boldsymbol{a}}, \boldsymbol{a} \right) + \left( 1 - \alpha \right) \mathcal{L}_1 \left( \hat{\boldsymbol{b}}, \boldsymbol{b} \right) , 
\end{equation}
where $\alpha \in \left( 0,1 \right)$ is a weight parameter that determines the importance of the training loss at real and imaginary parts, $\overline{\boldsymbol{W}}$ is the parameters of our designed model including all the weights and bias defined in (\ref{eq:conv1d}), (\ref{eq:fc}), and (\ref{eq:out}), $\boldsymbol{a},\boldsymbol{b} \in \mathbb{R}^{|\mathcal{D}_u| \times 1}$ are the vectors of the user's ground truth positions (i.e., $\boldsymbol{P} = \left[ \boldsymbol{a}, \boldsymbol{b} \right]$), and $\hat{\boldsymbol{a}}, \hat{\boldsymbol{b}} \in \mathbb{R}^{|\mathcal{D}_u| \times 1}$ are vectors of the real and the imaginary part of $\hat{\boldsymbol{y}}$ (i.e., $\hat{\boldsymbol{y}} = \hat{\boldsymbol{a}} + i\hat{\boldsymbol{b}}$), and $\mathcal{L}_1 \left( \hat{\boldsymbol{a}},\boldsymbol{a} \right)$ is the mean squared error (MSE) loss function that measures the difference between the predicted result $\hat{\boldsymbol{a}}$ and the ground truth result $\boldsymbol{a}$. MSE is defined as  
\begin{equation}\label{eq:mse}
    \mathcal{L}_1 \left( \hat{\boldsymbol{a}},\boldsymbol{a} \right) = \frac{1}{|\mathcal{D}_u|} \sum_{i=1}^{|\mathcal{D}_u|} \left( \hat{a}_i - a_i \right) ^2,  
\end{equation}
where $\hat{a}_i$ is element $i$ of $\hat{\boldsymbol{a}}$, and $a_i$ is element $i$ of $\boldsymbol{a}$. \\

\subsubsection{Loss Function for Use Case II}
In use case \textrm{II}, the output of our designed algorithm is two CSI features. In our proposed scheme, the real part $\hat{\boldsymbol{a}}$ is LOS/NLOS classification results and the imaginary part $\hat{\boldsymbol{b}}$ is the predictions of the signal TOA. Since LOS/NLOS classification is a binary classification task and signal TOA prediction is a regression task, we use different types of loss functions to measure the training loss of the real and the imaginary part. In particular, we use binary cross entropy loss function to measure the LOS/NLOS classification accuracy, and use MSE to measure signal TOA prediction accuracy. Then, the total loss of our designed model used for case \textrm{II} is
\begin{equation}\label{eq:loss2}
    J \left( \overline{\boldsymbol{W}}, \mathcal{D}_u, \boldsymbol{a}, \boldsymbol{b} \right) = \beta \mathcal{L}_2 \left( \hat{\boldsymbol{a}}, \boldsymbol{a} \right) + \left( 1 - \beta \right) \mathcal{L}_1 \left( \hat{\boldsymbol{b}}, \boldsymbol{b} \right), 
\end{equation}
where $\beta \in \left( 0,1 \right)$ is a weight parameter to adjust the importance of the loss at real and imaginary parts, $\boldsymbol{a}$ is a vector of the LOS/NLOS link labels, $\boldsymbol{b}$ is the vector of ground truth TOA of the signal, and $\mathcal{L}_2 \left( \hat{\boldsymbol{a}}, \boldsymbol{a} \right)$ is the binary cross entropy with respect to the LOS/NLOS classification result $\hat{\boldsymbol{a}}$ and the LOS/NLOS label $\boldsymbol{a}$. The binary cross entropy loss function is defined as 
\begin{equation}\label{eq:bce}
    \mathcal{L}_2 \left( \hat{\boldsymbol{a}}, \boldsymbol{a} \right) = -\frac{1}{|\mathcal{D}_u|} \sum_{i=1}^{|\mathcal{D}_u|} \boldsymbol{a} \log \left( \delta \left( \hat{\boldsymbol{a}} \right) \right) + \left( 1 - \boldsymbol{a}\right)\log \left( 1 - \delta \left( \hat{\boldsymbol{a}} \right) \right), 
\end{equation}
where $\delta \left( \cdot \right)$ is the sigmoid function. 
From (\ref{eq:loss2}), we see that the CVNN model can process two different types of learning tasks simultaneously. Therefore, compared to RVNNs that can process only one learning task per training, a CVNN model can use less neurons to implement more learning tasks thus reducing ML model training complexity and saving ML model training time. \\

\subsubsection{Training Process}
Given the defined loss functions, next, we introduce the training process of our designed FL algorithm so as to find the optimal model to solve problem (\ref{eq:problem}). The designed FL training process consists of two steps. In the first step, the users will use their local datasets to update their local FL models. Then, the devices will transmit their local FL model parameters to the server which aggregates the received local FL model parameters to generate a global model $\boldsymbol{g}$. Then, the global model $\boldsymbol{g}$ will be transmitted back to all users so that the users can update their local FL models continuously. Next, we introduce the local model update and global FL model update seperately. 
\begin{itemize}
    \item \textbf{Local Model Update}: First, we introduce the process of updating the local FL model $\overline{\boldsymbol{W}}_{u}^{t}$ of user $u$ at iteration $t$. We use a back-propagation algorithm with a mini-batch stochastic gradient descent (SGD) approach to update the local FL model $\overline{\boldsymbol{W}}_u$ of each user $u$ \cite{8264077}. The update of $\overline{\boldsymbol{W}}_u$ at iteration $t$ is: 
    \begin{equation}\label{eq:update_loc}
        \overline{\boldsymbol{W}}_{u}^{t+1} = \boldsymbol{g}_{t} - h\left( \eta, t \right) \frac{\partial J \left( \boldsymbol{g}_{t}, \mathcal{B}_{u}^{t} \right)}{\partial  \boldsymbol{g}_{t}^{*}}, 
    \end{equation}
    where $\mathcal{B}_u^{t} \subset \mathcal{D}_u$ is a batch of data samples of user $u$ at iteration $t$, $h\left( \eta, t \right)$ is the function of learning rate that is determined by the base learning rate $\eta$ and iteration $t$, and $ \boldsymbol{g}_{t}^{*}$ is the conjugate of $\boldsymbol{g}_{t}$. From (\ref{eq:update_loc}), we can see that the direction of gradient descent for a CVNN model is the derivative with respect to $\boldsymbol{g}_{t}^{*}$ instead of $\boldsymbol{g}_{t}$. Here, for each user $u$ at each iteration $t$, its local model can be updated more than once \cite{chen2021communication}. 
    \item \textbf{Global Model Update}: Next, we introduce the process of the global FL model update at the server. In our designed CVNN based FL method, we assume that each user may not transmit the entire complex-valued weight parameters to the server. In particular, we assume that each user can transmit real part or imaginary part of CVNN model to the server. Let $\mathfrak{R} \left( \overline{\boldsymbol{W}}_u^t \right)$ and $\mathfrak{I} \left( \overline{\boldsymbol{W}}_u^t \right)$ be the real and imaginary part of the CVNN model. Then, the process of the server aggregating the received local FL parameters of all the participating users into a global FL model is~\cite{mcmahan2017communication}:  
    \begin{equation}\label{eq:update_glb}
        \boldsymbol{g}_{t} = \frac{\sum_{u=1}^{U} r_u^t \mathfrak{R} \left( \overline{\boldsymbol{W}}_u^t \right)}{\sum_{u=1}^{U} |\mathcal{B}_{u}^{t}| r_u^t} + i \frac{\sum_{u=1}^{U} m_u^t \mathfrak{I} \left( \overline{\boldsymbol{W}}_u^t \right)}{\sum_{u=1}^{U} |\mathcal{B}_{u}^{t}| m_u^t}, 
    \end{equation}
    where $r_u^t \in \{ 0, 1 \}$ is used to indicate whether user $u$ transmits the real part of  $\overline{\boldsymbol{W}}_u^t$ to the server, and $m_u^t \in \{ 0, 1 \}$ is used to indicate whether user $u$ transmits the imaginary part of  $\overline{\boldsymbol{W}}_u^t$ to the server. More specifically, $r_u^t=1$ implies that user $u$ will transmit the real part of the local FL model $\overline{\boldsymbol{W}}_u^t$ to the server at FL iteration $t$ and $r_u^t=0$ otherwise. Similarly, $m_u^t=1$ implies that user $u$ will transmit the imaginary part of $\overline{\boldsymbol{W}}_u^t$ to the server at FL iteration $t$ and $m_u^t=0$ otherwise. 
\end{itemize}
The entire training process is described in \textbf{Algorithm \ref{alg:algorithm1}}. We first initialize the local FL model parameters $\overline{\boldsymbol{W}}_u^0$ for each user $u$. Then, we perform the FL training. At the first iteration (i.e., $t=1$), each user $u$ uses $\overline{\boldsymbol{W}}_u^0$ to update its local FL model. Otherwise, each user uses the global FL model $\boldsymbol{g}_t$ received from the server to update its local FL model. After $T$ training iterations, we can obtain a common FL model $\overline{\boldsymbol{g}}$. 

\begin{algorithm}[!t]
    \small
    \caption{\small The Training Process of the CVNN-based FL Algorithm}
    \label{alg:algorithm1}
    \begin{algorithmic}
        \REQUIRE local dataset of all $U$ users $\mathcal{D}_1,...,\mathcal{D}_U$; 
        \ENSURE $\overline{\boldsymbol{W}}_1^0,...,\overline{\boldsymbol{W}}_U^0$; 
        \FOR {$t = 1 \to T$}
            \STATE \textbf{Local model update at each device: } 
            \FOR {$u=1 \to U$}
                \STATE User $u$ uses $\mathcal{B}_{u}^t \subset \mathcal{D}_u$ to train the local FL model and obtain the prediction $\hat{\boldsymbol{y}}_{u}^t$; 
                \IF {$t=1$} 
                    \STATE User $u$ calculates the loss $J \left( \overline{\boldsymbol{W}}_u^0, \mathcal{B}_{u}^1 \right)$ based on (\ref{eq:loss1}) for case \uppercase\expandafter{\romannumeral1} or (\ref{eq:loss2}) for case \uppercase\expandafter{\romannumeral2}; 
                \ELSE
                    \STATE User $u$ calculates the loss $J \left( \boldsymbol{g}_{t-1}, \mathcal{B}_{u}^t \right)$ based on (\ref{eq:loss1}) for case \uppercase\expandafter{\romannumeral1} or (\ref{eq:loss2}) for case \uppercase\expandafter{\romannumeral2}; 
                \ENDIF
                \STATE User $u$ updates $\overline{\boldsymbol{W}}_{u}^t$ based on (\ref{eq:update_loc}) 
            \ENDFOR
            \STATE \textbf{Global model update at the server}: 
            \STATE The server updates $\boldsymbol{g}_t$ based on (\ref{eq:update_glb})
        \ENDFOR
    \end{algorithmic}
\end{algorithm}

\section{Convergence Analysis}
Next, we analyze the convergence and implementation of our proposed CVNN based FL. 

\subsection{Convergence Analysis of the Designed CVNN based FL}
We assume that $J_u \left( \boldsymbol{g}_t, \mathcal{B}_u^t \right)$ is the loss of user $u$ at iteration $t$, and $J \left( \boldsymbol{g}_t \right) = \frac{1}{N} \sum_{u=1}^U J_u \left( \boldsymbol{g}_t, \mathcal{B}_u^t \right)$ is the total loss of the FL algorithm at iteration $t$, with $N = \sum_{u=1}^U |\mathcal{B}_u^t|$. Given (\ref{eq:update_loc}) and (\ref{eq:update_glb}), the global FL model at iteration $t+1$ is updated by 
\begin{equation}\label{eq:update}
    \boldsymbol{g}_{t+1} = \boldsymbol{g}_t - h\left( \eta,t \right) \left( \nabla J \left( \boldsymbol{g}_t \right) - \boldsymbol{o} \right), 
\end{equation}
where $\boldsymbol{o} = \nabla J \left( \boldsymbol{g}_t \right) - \frac{\sum_{u=1}^U r_u^t \mathfrak{R}\left(\nabla J_u \left(\boldsymbol{g}_t\right)\right)}{\sum_{u=1}^U |\mathcal{B}_u^t| r_u^t} - i \frac{ \sum_{u=1}^U m_u^t \mathfrak{I}\left(\nabla J_u\left(\boldsymbol{g}_t\right)\right)}{\sum_{u=1}^U |\mathcal{B}_u^t| m_u^t}$. To analyze the convergence of the designed FL, we first make the following assumptions, as done in \cite{chen2020joint, amiri2021convergence}. 
\begin{assumption}
    \emph{We assume that the total loss function $J\left( \boldsymbol{g} \right)$ and the gradient $\nabla J\left( \boldsymbol{g} \right)$ of $J\left( \boldsymbol{g} \right)$ are complex-differentiable. }
\end{assumption}
\begin{assumption}\label{aspt:2}
    \emph{We assume that the gradient $\nabla J\left( \boldsymbol{g} \right)$ of the total loss $J\left( \boldsymbol{g} \right)$ is uniformly Lipschitz continuous with respect to the global FL model $\boldsymbol{g}$. Then, we have 
    \begin{equation}\label{eq:lip}
        \lVert \nabla J\left( \boldsymbol{g}_{t+1} \right) - \nabla J\left( \boldsymbol{g}_t \right) \rVert \leq Z \lVert \boldsymbol{g}_{t+1} - \boldsymbol{g}_t \rVert, 
    \end{equation}
    where $Z$ is a positive constant. }
\end{assumption}
\begin{assumption}
    \emph{We assume that $J\left( \boldsymbol{g} \right)$ is strongly convex with respect to a positive constant $\mu$. Then, we have 
    \begin{equation}\label{eq:convex}
        \begin{split}
            J\left( \boldsymbol{g}_{t+1} \right) \geq &\left( \boldsymbol{g}_t \right) + \left( \boldsymbol{g}_{t+1}-\boldsymbol{g}_t \right)^T \nabla J\left( \boldsymbol{g}_t \right) + \\
            &\frac{\mu}{2} \lVert \boldsymbol{g}_{t+1}-\boldsymbol{g}_t \rVert^2.
        \end{split} 
    \end{equation}}
\end{assumption}
\begin{assumption}
    \emph{We assume that 
    \begin{equation}\label{eq:bound1}
        \lVert \mathfrak{R} \left( \nabla J_u \left( \boldsymbol{g}_t, \boldsymbol{H}_{u,k}, \boldsymbol{p}_{u,k} \right) \right) \rVert^2 \leq \zeta_1 + \zeta_2 \lVert \nabla J \left( \boldsymbol{g}_t \right) \rVert^2, 
    \end{equation}
    \begin{equation}
        \lVert \mathfrak{I} \left( \nabla J_u \left( \boldsymbol{g}_t, \boldsymbol{H}_{u,k}, \boldsymbol{p}_{u,k} \right) \right) \rVert^2 \leq \zeta_1 + \zeta_2 \lVert \nabla J \left( \boldsymbol{g}_t \right) \rVert^2, 
    \end{equation}
    with $\boldsymbol{H}_{u,k}, \boldsymbol{p}_{u,k}$ being the components of the data sample $k$ in $\mathcal{B}_u^t$, and $\zeta_1, \zeta_2 \geq 0$. }
\end{assumption}

Given these assumptions, the convergence of our designed FL is analyzed in the following theorem. 

\begin{theorem}\label{th:converge}
    \emph{Given the transmission indicators $\boldsymbol{r}^t$ and $\boldsymbol{m}^t$, the optimal global FL model $\overline{\boldsymbol{g}}$, and the learning rate $h \left( \eta, t \right) = \frac{1}{Z}$, the upper bound of $\mathbb{E} \left( J(\boldsymbol{g}_{t+1}) - J(\overline{\boldsymbol{g}}) \right)$ can given by
    \begin{equation}\label{eq:converge}
    \begin{split}
        \mathbb{E} &\left( J(\boldsymbol{g}_{t+1}) - J(\overline{\boldsymbol{g}}) \right) \\
        & \leq A^{t}  \mathbb{E} \left( J(\boldsymbol{g}_{1}) - J(\overline{\boldsymbol{g}}) \right) + \left( \frac{1-A^{t-1}}{1-A} \right) \frac{2 \zeta_1 E}{ZN}, 
    \end{split}
    \end{equation}
    where $A = 1-\frac{\mu}{Z}+\frac{4\mu\zeta_2E}{NZ}$, and $E = 2N - \mathbb{E} \left( \sum_{u=1}^U |\mathcal{B}_u^t| r_u^t + \sum_{u=1}^U |\mathcal{B}_u^t| m_u^t \right)$ with $\mathbb{E} \left( \sum_{u=1}^U |\mathcal{B}_u^t| r_u^t \right)$ being the expected total training samples of the users that send $\mathfrak{R} \left( \overline{\boldsymbol{W}}_u^t \right)$ to the server}, and $\mathbb{E} \left( \sum_{u=1}^U |\mathcal{B}_u^t| m_u^t \right)$ being the expected total training samples of the users that send $\mathfrak{I} \left( \overline{\boldsymbol{W}}_u^t \right)$ to the server. 
\end{theorem}
\begin{proof}
   See Appendix A. 
\end{proof}
In Theorem \ref{th:converge}, $\boldsymbol{g}_{t+1}$ is the global FL model that is generated based on the real and imaginary parts of FL models transmitted by the users at iteration $t+1$. From Theorem \ref{th:converge}, we can see that a gap, $\left( \frac{1-A^{t}}{1-A} \right) \frac{2 \zeta_1 E}{ZN}$, exists between $\mathbb{E} \left( J(\boldsymbol{g}_{t+1}) \right)$ and $\mathbb{E} \left( J(\overline{\boldsymbol{g}}) \right)$. The gap is caused by the policy of real part and imaginary part of FL model transmission. When the number of users that transmit real or imaginary part of FL models increases, the value of $A$ decreases, and thus the gap $\left( \frac{1-A^{t}}{1-A} \right) \frac{2 \zeta_1 E}{ZN}$ decreases and the convergence speed of FL increases. Based on Theorem \ref{th:converge}, we can next derive the convergence rate of our designed FL algorithm when all users send their complete local FL models (i.e., both real and imaginary parts) to the server at all iterations. 
\begin{lemma}\label{pro:converge}
    \emph{Given the optimal FL model $\overline{\boldsymbol{g}}$, the learning rate $h \left( \eta, t \right) = \frac{1}{Z}$, and $r_u^t = m_u^t = 1$ for each user $u$, the upper bound of $\mathbb{E} \left( J(\boldsymbol{g}_{t+1}) - J(\overline{\boldsymbol{g}}) \right)$ is given by
    \begin{equation}\label{eq:converge_lemma}
        \mathbb{E} \left( J(\boldsymbol{g}_{t+1}) - J(\overline{\boldsymbol{g}}) \right) \leq \left( 1 - \frac{\mu}{Z}\right)^t \mathbb{E} \left( J(\boldsymbol{g}_{1}) - J(\overline{\boldsymbol{g}}) \right). 
    \end{equation}} 
\end{lemma}
\begin{proof}
    Since all users send their complex-valued local FL models to the server at each iteration $t$, we have $\mathbb{E} \left( r_u^t \right) = 1$, $\mathbb{E} \left( m_u^t \right) = 1$, $E = 2N - \mathbb{E} \left( \sum_{u=1}^U |\mathcal{B}_u^t| r_u^t + \sum_{u=1}^U |\mathcal{B}_u^t| m_u^t \right) = 2N - 2 \sum_{u=1}^U |\mathcal{B}_u^t| = 0$. Since $E=0$, $A = 1 - \frac{\mu}{Z}$ and $\left( \frac{1-A^{t}}{1-A} \right) \frac{2 \zeta_1 E}{ZN} = 0$. Then, we substitute $A=0$ into (\ref{eq:converge}) to obtain (\ref{eq:converge_lemma}). This completes the proof. 
\end{proof}
From Proposition \ref{pro:converge}, we can see that, when all complete local FL models are sent to the server, our designed FL model will converge to the globally optimal. 

\subsection{Implementation and Complexity}
Here, we first analyze the implementation of our designed CVNN based FL algorithm. The implementation of the designed FL consists of local FL model update and global FL model update. For local FL model update, each user $u$ must collect a local CSI dataset $\mathcal{D}_u$. To update a local FL model at iteration $t$, each user $u$ must select a batch of data $\mathcal{B}_u^{t}$ from the local CSI dataset $\mathcal{D}_u$. Additionally, each user $u$ must receive the global FL model $\boldsymbol{g}_t$ from the server. For global FL model update at iteration $t$, the server must first receive the local FL model of each user $u$, the indicators $r_u^t$ and $m_u^t$, and the size of the training batch $|\mathcal{B}_u^t|$. \\

We next analyze the complexity of our designed algorithm. The time complexity of a local FL model can be evaluated by the number of multiplication operations. According to (\ref{eq:conv1d}), the time complexity of convolutional layers are respectively $\mathcal{O} \left( O^\textrm{I} S_1^2 C L \right)$ and $\mathcal{O} \left( O^\textrm{I} O^\textrm{II} S_2^2 C^\textrm{I} \overline{L}^\textrm{I} \right)$. The time complexity of pooling layers are respectively $\mathcal{O} \left( O^\textrm{I} C^\textrm{I} L^\textrm{I} \right)$ and $\mathcal{O} \left( O^\textrm{II} C^\textrm{II} L^\textrm{II} \right)$. From (\ref{eq:fc}), the time complexity of fully connected layer \textrm{I}, fully connected layer \textrm{II}, and the output layer are respectively  $\mathcal{O} \left( C^{\textrm{II}} \overline{L}^{\textrm{II}} N^{\textrm{I}} \right)$, $\mathcal{O} \left( N^{\textrm{I}} N^{\textrm{II}}\right)$, and $\mathcal{O} \left( N^{\textrm{II}} \right)$. Thus, the total time complexity of a local FL model is \cite{qiu2016going} 
$ \mathcal{O}  \left( O^\textrm{I} S_1^2 C L + O^\textrm{I} O^\textrm{II} S_2^2 C^\textrm{I} \overline{L}^\textrm{I} + O^\textrm{I} C^\textrm{I} L^\textrm{I} + O^\textrm{II} C^\textrm{II} L^\textrm{II} \right. + \left. C^{\textrm{II}} \overline{L}^{\textrm{II}} N^{\textrm{I}} + N^{\textrm{I}} N^{\textrm{II}} + N^{\textrm{II}} \right) \approx \mathcal{O} \left( O^\textrm{I} O^\textrm{II} S_2^2 C^\textrm{I} \overline{L}^\textrm{I} \right). $

The space complexity of the model refers to the memory footprint. Given the introduction of components of our designed local FL model, for each user $u$, the space complexity of convolutional layers are respectively $\mathcal{O} \left( C C^{\textrm{I}} S_1^2 \right)$ and $\mathcal{O} \left( C^{\textrm{I}} C^{\textrm{II}} S_2^2 \right)$. The space complexity of fully connected layer \textrm{I}, fully connected layers are respectively $\mathcal{O} \left( C^{\textrm{II}} \overline{L}^{\textrm{II}} N^{\textrm{I}} \right)$, $\mathcal{O} \left( + N^{\textrm{I}} N^{\textrm{II}} \right)$, and $\mathcal{O} \left( N^{\textrm{II}} \right)$. Thus, the total space complexity of a local FL model is 
    $\mathcal{O} \left( C C^{\textrm{I}} S_1^2 + C^{\textrm{I}} C^{\textrm{II}} S_2^2                            
           + C^{\textrm{II}} \overline{L}^{\textrm{II}} N^{\textrm{I}} + N^{\textrm{I}} N^{\textrm{II}} + N^{\textrm{II}} \right)  \approx \mathcal{O} \left( C^{\textrm{I}} C^{\textrm{II}} S_2^2 \right). $

\section{Simulation Results}\label{se:simulation}
In this section, we perform extensive simulations to evaluate the performance of our designed CVNN based FL in two specific scenarios: 1) the output of our designed algorithm is the estimated positions of users, 2) the output of our designed algorithm is two CSI feature which can be used for traditional positioning methods. We first introduce the CSI dataset used to train the designed CVNN model. Then, we explain the parameters of our proposed CVNN model and a RVNN model based baseline. Finally, we analyze the simulation results of our designed CVNN model. Note that, in Figs. \ref{fig:mse_set2_pos}, \ref{fig:real_format}, \ref{fig:mse_set2_los}, \ref{fig:mse_set2_los_user}, and \ref{fig:mse_set1_pos} we have removed the initial epochs where the loss is very large so as to clearly show the gap of the loss between our designed CVNN based method and the baseline RVNN based method when the considered algorithms converge. 

\subsection{Dataset Introduction}
\subsubsection{5G CSI Dataset} The first CSI dataset we use to evaluate our designed CVNN based FL algorithm is from \cite{pan2023situ}. At each position, 100 CSI data are collected over 4 antennas and 1632 subcarriers. In our simulation, we only use the CSI data collected by antennas 1 and 2 (i.e., $C$ = 2). At each antenna, we transfer the CSI data from the frequency domain to the time domain by the inverse Fourier transform. For simplicity, we use only the first 250 sampling intervals and hence $L = 250$. Thus, in our simulations, we have 47600 CSI samples in total. We assign 42840 samples to each user equally such that each user has 3570 data samples. \\

\subsubsection{Cellular Ultra Dense CSI Dataset} The CSI dataset in \cite{9535488} is used to train our designed CVNN based FL model. The position of the server and the moving areas of $U$ users.
In \cite{9535488}, the server equipped with 64 antennas collects CSI data using three different antenna array topologies: 1) a uniform linear array (ULA) of 1 $\times$ 64 antennas, 2) a uniform rectangular array (URA) of 8 $\times$ 8 antennas, and 3) eight distributed ULAs of 1 $\times$ 8 antennas. In our simulations, we use the data collected by the antennas with URA topology. For simplicity, we use only the CSI data collected by 2 antennas (i.e., $C = 2$) and the position coordinate of the antennas is $\left[ -175,0 \right]$. Each CSI signal is collected over 100 sampling intervals and hence $L = 100$. Each antenna collects 264001 data samples and each data sample consists of CSI, position coordinate of the user, and the label of LOS/NLOS signal transmission link. Since the time slots of two successive data samples are very close, we only take one sample from every 10 samples. Hence, in our simulations, we use 25201 data samples and assign 22680 samples to all users equally such that each user has 1890 data samples. For different use cases, we use the same CSI matrix as the input of our designed FL algorithm while the labels are different. In particular, for use case \textrm{I}, the output is the user's position coordinate $\boldsymbol{p}$. For use case \textrm{II}, the output is the distance between the user and the server, and LOS/NLOS link classification result. 

\subsection{CVNN Based FL Algorithm Parameter Introduction}\label{ss:model_param}
The parameters of the designed CVNN based FL Algorithm are summarized in Table \ref{table:systemparameters}. The function of learning rate $h \left( \eta, t \right)$ is 
 \begin{equation}\label{eq:eta}
    h \left( \eta, t \right) = \left\{
            \begin{aligned}
                &\eta && t \leq 50 , \\
                &\frac{1}{5} \eta && 50 < t \leq 75 ,\\
                &\frac{1}{2} \eta && t > 75 . \\
            \end{aligned}
            \right.
\end{equation}
 For comparison purposes, we use a RVNN based local FL model as the baseline. The baseline model parameters are similar to the CVNN based local FL model. We separate the real part and the imaginary part of each CSI sample of the dataset into two matrices $\mathfrak{R} \left( \boldsymbol{H} \right)$ and $\mathfrak{I} \left( \boldsymbol{H} \right)$. Then, the input of the RVNN is $\left[ \mathfrak{R} \left( \boldsymbol{H} \right),\mathfrak{I} \left( \boldsymbol{H} \right) \right]$, and the output is $\left[ \mathfrak{R} \left( \hat{y} \right), \mathfrak{I} \left( \hat{y} \right) \right]$. We can see that the input layer and the output layer of the RVNN is double of the CVNN model. Note that, the weight matrices and bias of the RVNN are all real-valued. 
\begin{table}[!t]
    \caption{System Parameters}
    \label{table:systemparameters}
    \centering
    \begin{tabular}{|c|c|c|c|}
        \hline
        \textbf{Parameter} & \textbf{Value} & \textbf{Parameter} & \textbf{Value} \\
        \hline
        $T$ & 85 & $| \mathcal{B}_u^t |$ & 32 \\
        \hline
        $\eta$ & $1 \times 10^{-4}$ & $O^{\textrm{I}}$ & 4 \\
        \hline
        $S_1$ & 2 & $P^{\textrm{I}}$ & 5 \\ 
        \hline
        $S^{\textrm{I}}$ & 1 & $O^{\textrm{II}}$ & 8 \\
        \hline
        $S_2$ & 2 & $P^{\textrm{II}}$ & 9 \\
        \hline
        $S^\textrm{II}$ & 2 & $N^{\textrm{I}}$ & 64 \\
        \hline
        $N^{\textrm{II}}$ & 32 & &  \\
        \hline
    \end{tabular}
\end{table}  

\begin{figure}[!t]
  \begin{center}
    \includegraphics[width=7.5cm]{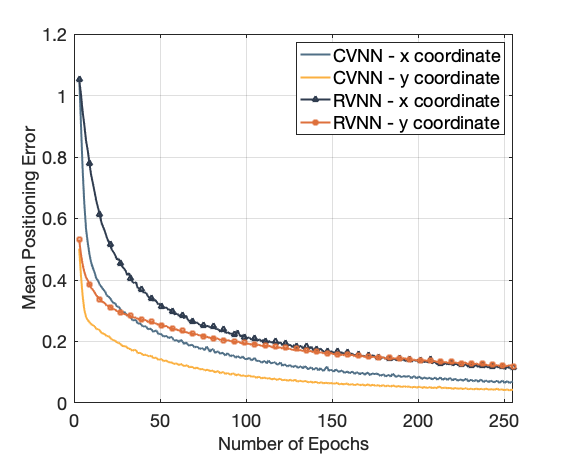}
    \caption{The training loss changes as the number of training iterations varies for use case \uppercase\expandafter{\romannumeral1} of the 5G CSI dataset. } 
    \label{fig:mse_set2_pos}
    \vspace{-0.4cm}
  \end{center}
\end{figure}

\subsection{Simulation Results of the 5G CSI Dataset}
In Fig. \ref{fig:mse_set2_pos}, we show how the value of the positioning error defined in (\ref{eq:problem}) changes as the number of training iterations varies. Fig. \ref{fig:mse_set2_pos} shows that as the number of iterations increases, the mean positioning errors of both considered algorithms decreases. This is because the models in the considered algorithms are updated by the CSI data at each iteration. From Fig. \ref{fig:mse_set2_pos}, we also see that our designed FL algorithm can achieve up to 36\% gain in terms of mean positioning error compared to the RVNN baseline. This is due to the fact that the CVNN model can directly process complex-valued CSI data without any data transformation thus obtaining more CSI features. \\

Fig. \ref{fig:cdf_set2_pos} shows the cumulative distribution function (CDF) of the positioning error resulting from our designed algorithm and the RVNN baseline. From Fig. \ref{fig:cdf_set2_pos} we see that, compare to the RVNN baseline, our designed algorithm improves the CDF of up to 33\% gains at a positioning error of 0.2 compared to the RVNN baseline. This is because our designed model does not need to preprocess complex-valued CSI data, thus it can obtain more CSI features compared to the RVNN baseline \cite{8297024}. \\
\begin{figure}[t]
  \begin{center}
    \includegraphics[width=7.5cm]{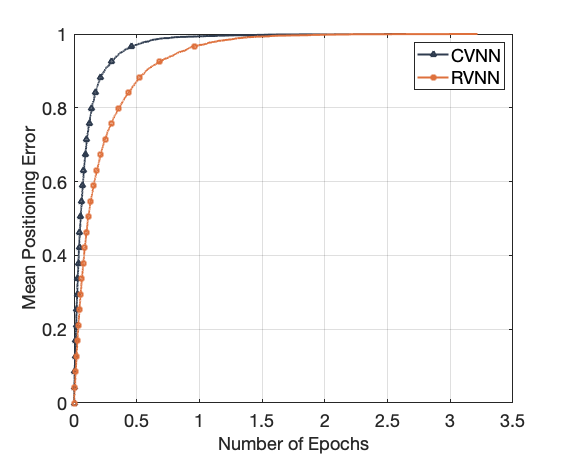}
    \caption{CDF of positioning MSE for use case \textrm{I} of the 5G CSI dataset. } 
    \label{fig:cdf_set2_pos}
    \vspace{-0.5cm}
  \end{center}
\end{figure}

\begin{figure}[t]
    \centering
    \subfigure[x coordinate] {\includegraphics[width=7.5cm]{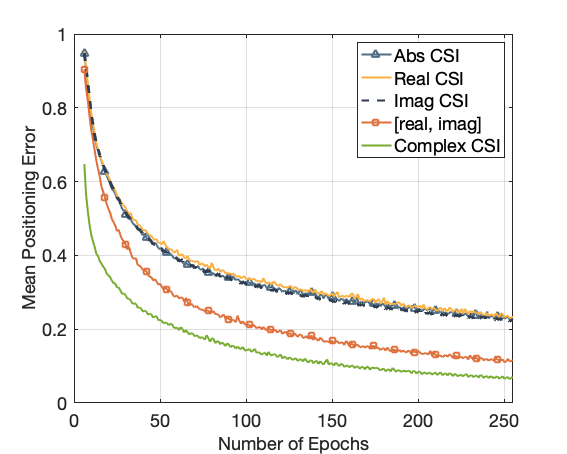}}
    \subfigure[y coordinate] {\includegraphics[width=7.5cm]{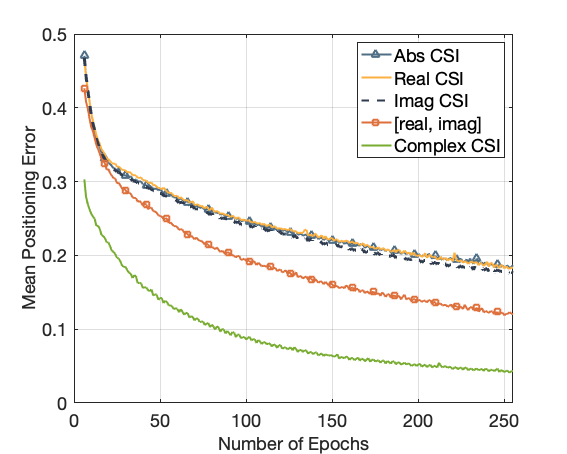}}
    \caption{The training loss changes as the number of training iterations varies for use case \uppercase\expandafter{\romannumeral1} of the 5G CSI dataset. } 
    \label{fig:real_format}
    \vspace{-0.5cm}
\end{figure}

\begin{figure}[t]
  \begin{center}
    \includegraphics[width=7.5cm]{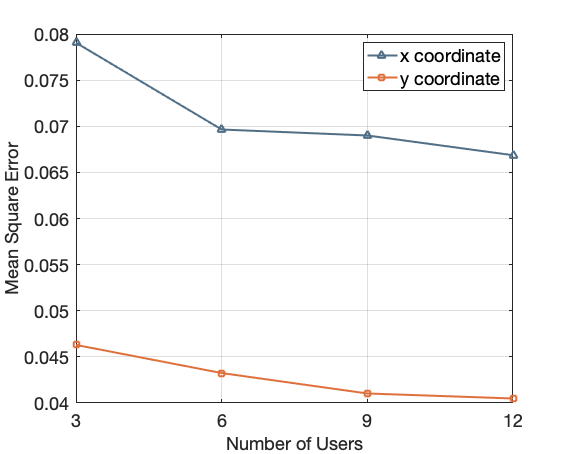}
    \caption{The optimal training loss changes as the number of users varies for use case \textrm{I} of the 5G CSI dataset.  } 
    \label{fig:usr_num}
    \vspace{-.5cm}
  \end{center}
\end{figure}

\begin{figure}[t]
  \begin{center}
    \includegraphics[width=7.5cm]{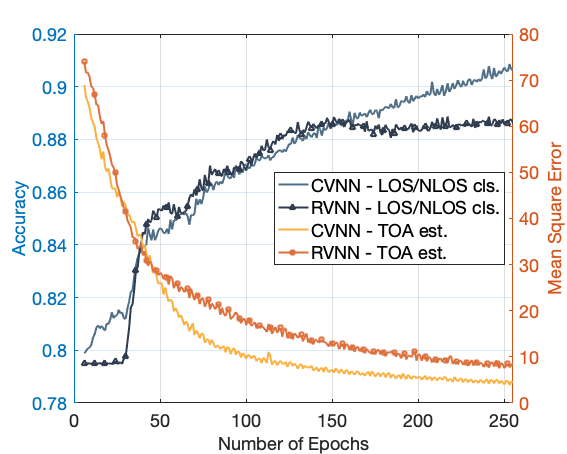}
    \caption{The mean TOA estimation error and the mean accuracy of the LOS/NLOS classification change as the number of training iterations varies for use case \uppercase\expandafter{\romannumeral2} of the 5G CSI dataset. } 
    \label{fig:mse_set2_los}
    \vspace{-.3cm}
  \end{center}
\end{figure}

\begin{figure*}[h]
    \centering
    \subfigure[User A] {\includegraphics[width=.32\linewidth]{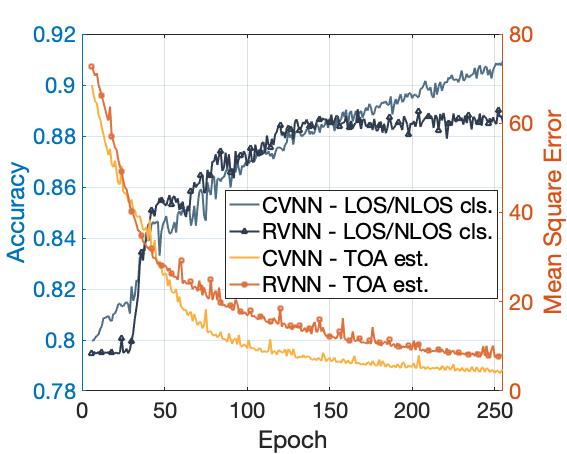}}
    \subfigure[User B] {\includegraphics[width=.32\linewidth]{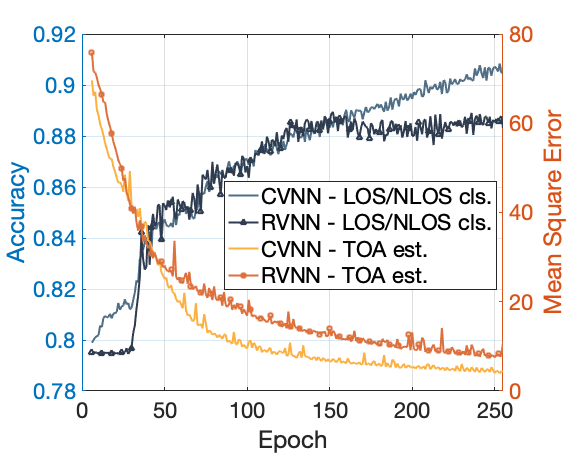}}
    \subfigure[User C] {\includegraphics[width=.32\linewidth]{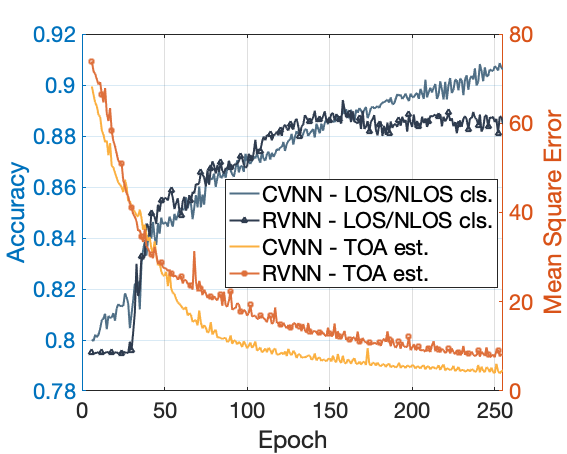}}
    \caption{The value of the TOA estimation error and the LOS/NLOS classification accuracy of three users change as the training iterations varies for use case \textrm{II} of the 5G CSI dataset. } 
    \label{fig:mse_set2_los_user}
    \vspace{-0.3cm}
\end{figure*}

In Fig. \ref{fig:real_format}, we show how the value of the positioning error changes as the number of training iterations varies when the CSI data is transformed into real-valued data via different methods. Here, we consider the use of four methods to process CSI data: 1) Using the real part of the CSI (i.e., $\mathfrak{R} \left( \boldsymbol{H} \right)$) and ignore the imaginary part of the CSI data, 2) using the imaginary part of the CSI (i.e., $\mathfrak{I} \left( \boldsymbol{H} \right)$) and ignore the real part of the CSI data, 3) using the absolute value of the CSI (i.e., $\left| \boldsymbol{H} \right|$), and 4) using the RVNN baseline that is described in Section \ref{ss:model_param}. From Fig. \ref{fig:real_format}, we see that the positions of users can be estimated even when we use the real or the imaginary part of CSI data as input. This is because both real and imaginary parts of CSI data contain positioning information. Fig. \ref{fig:real_format} also shows the RVNN baseline can achieve up to 49.35\%, 48.55\%, and 50.50\% gains in terms of the mean positioning error compared to the RVNNs trained by the absolute value of the CSI, the imaginary part of CSI, and the real part of CSI. This is due to the fact that the RVNN baseline uses both real and imaginary parts of CSI data for user positioning. However, the RVNN baseline doubles the size of the input vector of the ML model, which may significantly increase the training complexity of the ML model.

In Fig. \ref{fig:usr_num}, we show how the value of the positioning error changes as the number of training iterations varies. From Fig. \ref{fig:usr_num}, we see that as $U$ increases, the mean positioning error decreases. This is because when more users participate in the FL training, the total number of training samples used for training FL models increases.

Fig. \ref{fig:mse_set2_los} shows the mean TOA estimation error and the mean accuracy of the LOS/NLOS classification change as the number of training iterations varies. From Fig. \ref{fig:mse_set2_los}, we see that our designed algorithm can achieve up to 53.28\% gain in terms of the TOA estimation error, and 1.44\% gain in terms of the LOS/NLOS classification accuracy. This is due to the fact that the CVNN has a better generalization ability to process complex-valued data compared to the RVNN.

\begin{figure}[t]
  \begin{center}
  \includegraphics[width=7.5cm]{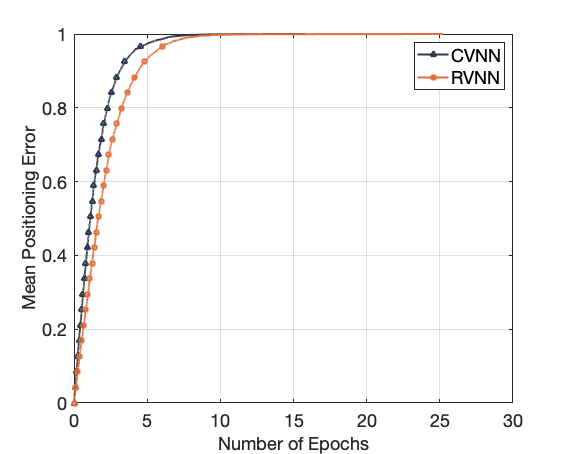}
  \caption{CDF of the TOA estimation error for use case \textrm{II} of the 5G CSI dataset. } 
  \label{fig:cdf_set2_los}
  \vspace{-.5cm}
  \end{center}
\end{figure}

\begin{figure}[t]
  \begin{center}
    \includegraphics[width=7.5cm]{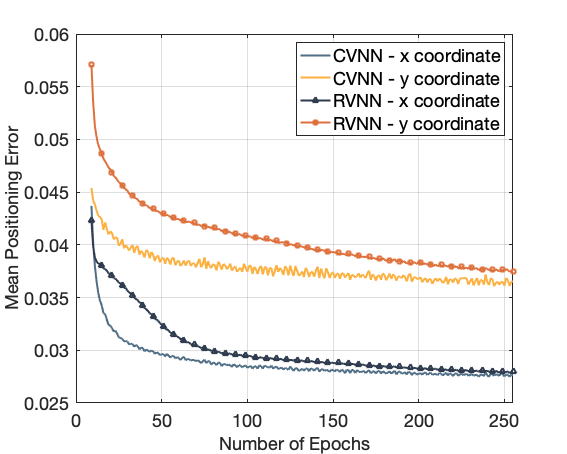}
    \caption{The value of the positioning error changes as the number of training iterations varies for use case \uppercase\expandafter{\romannumeral1} of the Cellular Ultra Dense CSI dataset. } 
    \label{fig:mse_set1_pos}
    \vspace{-0.5cm}
  \end{center}
\end{figure}


In Fig. \ref{fig:mse_set2_los_user}, we show how the value of the positioning errors, the value of the TOA estimation error, and the LOS/NLOS classification accuracy of three users change as the training iterations varies. These three users are randomly selected from 6 users. From Fig. \ref{fig:mse_set2_los_user}, we see that three users have different mean positioning errors. This is because the local FL model of each user is trained by its local dataset, and different users have different local datasets. \\

Fig. \ref{fig:cdf_set2_los} shows the CDF of the TOA estimation error resulting from our designed algorithm and the RVNN baseline. From Fig. \ref{fig:cdf_set2_los} we see that, compare to the RVNN baseline, our designed algorithm improves the CDF of up to 16.14\% gains at a positioning error of 3 compared to the RVNN baseline. This is because the complex-valued activation function defined in (\ref{eq:phi1}) can reduce redundant information of the training CSI data and help the model learn a more sparse representation~\cite{karnewar2022relu}.

\subsection{Simulation Results of the Cellular Ultra Dense CSI Dataset}
In Fig. \ref{fig:mse_set1_pos}, we show how the value of the positioning error defined in (\ref{eq:problem}) changes as the number of training iterations varies. Fig. \ref{fig:mse_set1_pos} shows that our designed FL algorithm can achieve up to 1.01\% and 5.68\% gains in terms of x and y coordinates mean positioning error compared to the RVNN baseline. This is because the designed CVNN process complex-valued CSI directly without separating the complex numbers into real and imaginary parts. Therefore, the CVNN can better capture the relationship between the real and imaginary part of the complex-valued CSI compared to RVNN. 



\section{Conclusion}
In this paper, we have designed a novel indoor multi-user positioning system. We have formulated this indoor positioning problem as an optimization problem whose goal is to minimize the gap between the estimated position and the actual position. To solve this problem, we have proposed a CVNN-based FL algorithm that has two key advantages: 1) our proposed algorithm can directly process complex-valued CSI data without data transformation, and 2) our proposed algorithm is a distributed ML method that does not require users to send their CSI data to the server. Since the output of our proposed algorithm is complex-valued which consists of the real and imaginary parts, we can use it to implement two learning tasks. First, the proposed algorithm directly outputs the estimated positions users. Here, the real and imaginary parts of an output neuron represent the 2D coordinates of the user. Second, the proposed algorithm can output two CSI features. Simulation results have shown that the proposed CVNN-based FL algorithm yields significant improvements in the performance compared to a RVNN baseline which has to transform the complex-valued CSI data into real-valued data. \\

\section*{Appendix}
\subsection{Proof of Theorem \ref{th:converge}}\label{Ap:a}
To prove Theorem \ref{th:converge}, we first expand $J\left( \boldsymbol{g}_{t+1} \right)$ by using the second-order Taylor expansion, as follows: 
\begin{equation}\label{eq:taylor}
    \begin{split}
        J\left( \boldsymbol{g}_{t+1} \right) &= J\left( \boldsymbol{g}_t \right) + \left( \boldsymbol{g}_{t+1}-\boldsymbol{g}_t \right)^T \nabla J\left( \boldsymbol{g}_t \right) \\ 
        & \quad + \frac{1}{2} \left( \boldsymbol{g}_{t+1}-\boldsymbol{g}_t \right)^T \nabla^2 J\left( \boldsymbol{g}_t \right) \left( \boldsymbol{g}_{t+1}-\boldsymbol{g}_t \right) \\
        & \leq J \left( \boldsymbol{g}_t \right) + \left( \boldsymbol{g}_{t+1}-\boldsymbol{g}_t \right)^T \nabla J\left( \boldsymbol{g}_t \right) \\
        & \quad + \frac{Z}{2} \lVert \boldsymbol{g}_{t+1}-\boldsymbol{g}_t \rVert^2, 
    \end{split}
\end{equation}
where the inequality stems from the fact that $\nabla^2 J\left( \boldsymbol{g}_t \right) = \lim_{\boldsymbol{g}_{t+1} \to \boldsymbol{g}_t} \frac{\nabla J\left( \boldsymbol{g}_{t+1} \right) - \nabla J\left( \boldsymbol{g}_t \right)}{\boldsymbol{g}_{t+1}-\boldsymbol{g}_t} \leq Z$ which can be derived from Assumption \ref{aspt:2}. In (\ref{eq:taylor}), since $r_u^t$ and $m_u^t$ are random variables, based on the update policy in (\ref{eq:update_glb}), the global FL model $\boldsymbol{g}_{t+1}$ is a random variable. Thus, the loss $J \left( \boldsymbol{g}_{t+1} \right)$ is a random variable. To this end, we calculate the expectation of $J \left( \boldsymbol{g}_{t+1} \right)$ with respect to $r_u^t$ and $m_u^t$. Given $h \left( \eta, t \right) = \frac{1}{Z}$ and (\ref{eq:update}), the expectation of $J\left( \boldsymbol{g}_{t+1} \right)$ with respect to $r_u^t$ and $m_u^t$ is 
\begin{equation}\label{prf:j_1}
    \begin{split}
        \mathbb{E} \left( J(\boldsymbol{g}_{t+1}) \right) &\leq \mathbb{E} \left[ J \left( \boldsymbol{g}_t \right) - \frac{1}{Z} \left( \nabla J \left( \boldsymbol{g}_t \right) - \boldsymbol{o} \right)^T \nabla J \left( \boldsymbol{g}_t \right) \right. \\
        &\quad \left. + \frac{Z}{2} \frac{1}{Z^2} \lVert \nabla J \left( \boldsymbol{g}_t \right) - \boldsymbol{o} \rVert^2 \right] \\
        &= \mathbb{E} \left[ J \left( \boldsymbol{g}_t \right) \right] - \left( \frac{\lVert \nabla J \left( \boldsymbol{g}_t \right) \rVert^2}{Z} - \frac{\boldsymbol{o}^T \nabla J \left( \boldsymbol{g}_t \right)}{Z} \right) \\
        & \quad + \frac{1}{2Z} \mathbb{E} \left[ \lVert \nabla J \left( \boldsymbol{g}_t \right) \rVert^2 + \lVert \boldsymbol{o} \rVert^2 - 2 \boldsymbol{o}^T \nabla J \left( \boldsymbol{g}_t \right) \right] \\
        & = \mathbb{E} \left[ J \left( \boldsymbol{g}_t \right) \right] - \frac{1}{2Z} \lVert \nabla J \left( \boldsymbol{g}_t  \right) \rVert^2 + \frac{1}{2Z} \mathbb{E} \left[ \lVert \boldsymbol{o} \rVert^2 \right],  
    \end{split}
\end{equation}
where $\boldsymbol{o} = \nabla J \left( \boldsymbol{g}_t \right) - \frac{\sum_{u=1}^U r_u^t \mathfrak{R}\left(\nabla J_u \left(\boldsymbol{g}_t\right)\right)}{\sum_{u=1}^U |\mathcal{B}_u^t| r_u^t} - \frac{\sum_{u=1}^U m_u^t \mathfrak{I}\left(\nabla J_u \left(\boldsymbol{g}_t\right)\right)}{\sum_{u=1}^U |\mathcal{B}_u^t| m_u^t} $. To prove the convergence of our proposed method, we need to prove that the difference between the loss of the model $\boldsymbol{g}_t$ and the loss of the optimal model $\overline{\boldsymbol{g}}$, i.e., $J ( \boldsymbol{g}_t ) - J ( \overline{\boldsymbol{g}} )$, has an upper bound. To this end, we next simplify (\ref{prf:j_1}) by simplifying $\mathbb{E} \left[ \lVert \boldsymbol{o} \rVert^2 \right]$. Given $\boldsymbol{o} = \mathfrak{R} \left( \boldsymbol{o} \right) + i \mathfrak{I} \left( \boldsymbol{o} \right)$, we have $\mathbb{E} \left[ \lVert \boldsymbol{o} \rVert^2 \right]= \mathbb{E} \left[ \lVert \mathfrak{R} \left( \boldsymbol{o} \right) \rVert^2 \right] + \mathbb{E} \left[ \lVert \mathfrak{I} \left( \boldsymbol{o} \right) \rVert^2 \right]$. $\mathbb{E} \left[ \lVert \mathfrak{R} \left( \boldsymbol{o} \right) \rVert^2 \right]$ can be rewritten as follows 
\begin{equation}
    \begin{split}
        &\mathbb{E} \left[ \lVert \mathfrak{R} \left( \boldsymbol{o} \right) \rVert^2 \right] \\
        &= \mathbb{E} \left[ \lVert \mathfrak{R} \left( \nabla J \left( \boldsymbol{g}_t \right) \right) -\frac{\sum_{u=1}^U r_u^t \mathfrak{R}\left(\nabla J_u \left(\boldsymbol{g}_t\right)\right)}{\sum_{u=1}^U |\mathcal{B}_u^t| r_u^t} \rVert^2 \right].
\end{split}
\end{equation}
Since $\mathfrak{R} \left( \nabla J \left( \boldsymbol{g}_t \right) \right) = \frac{1}{N} \sum_{u=1}^U \mathfrak{R} \left( \nabla J_u \left( \boldsymbol{g}_t \right) \right)$, we have
\begin{equation}
    \begin{split}
     &\mathbb{E} \left[ \lVert \mathfrak{R} \left( \boldsymbol{o} \right) \rVert^2 \right]\\
        &= \mathbb{E} \left[ \lVert \frac{\sum_{u=1}^U \mathfrak{R} \left( \nabla J_u \left( \boldsymbol{g}_t \right) \right)}{N} - \frac{\sum_{u=1}^U r_u^t \mathfrak{R}\left(\nabla J_u \left(\boldsymbol{g}_t\right)\right)}{\sum_{u=1}^U |\mathcal{B}_u^t| r_u^t} \rVert^2 \right] \\ 
        &= \mathbb{E} \left[ \lVert - \frac{\left( N - \sum_{u=1}^U |\mathcal{B}_u^t| r_u^t\right) \sum_{u \in \mathcal{R}_1^t} \mathfrak{R} \left( \nabla J_u \left( \boldsymbol{g}_t \right) \right)}{N \sum_{u=1}^U |\mathcal{B}_u^t| r_u^t} \right. \\
        & \quad \left. + \frac{\sum_{u \in \mathcal{R}_2^t} \mathfrak{R} \left( \nabla J_u \left( \boldsymbol{g}_t \right) \right)}{N} \rVert^2 \right].\\ 
        & = \mathbb{E} \left[ \lVert \frac{ \left( N - \sum_{u=1}^U |\mathcal{B}_u^t| r_u^t\right) \sum_{u \in \mathcal{R}_1^t} \mathfrak{R} \left( \nabla J_u \left( \boldsymbol{g}_t \right) \right)}{N \sum_{u=1}^U |\mathcal{B}_u^t| r_u^t} \rVert^2 \right. \\
        & \quad \left. + \lVert \frac{\sum_{u \in \mathcal{R}_2^t} \mathfrak{R} \left( \nabla J_u \left( \boldsymbol{g}_t \right) \right)}{N} \rVert^2 - \frac{2\left( N - \sum_{u=1}^U |\mathcal{B}_u^t| r_u^t \right)}{N^2 \sum_{u=1}^U |\mathcal{B}_u^t| r_u^t} \right. \\
        & \quad \left. \langle \sum_{u \in \mathcal{R}_1^t} \mathfrak{R} \left( \nabla J_u \left( \boldsymbol{g}_t \right) \right), \sum_{u \in \mathcal{R}_2^t} \mathfrak{R} \left( \nabla J_u \left( \boldsymbol{g}_t \right) \right) \rangle \right].
\end{split}
\end{equation}
where $\mathcal{R}_1^t = \{u \in \mathcal{U} | r_u^t = 1 \}$ is the set of users that transmit the real parts of their local FL models to the server at iteration $t$, and $\mathcal{R}_2^t = \{ u \in \mathcal{U} | u \notin \mathcal{R}_1 \}$ is the set of users that do not transmit the real parts of their local FL models to the server. Since $\frac{2\left( N - \sum_{u=1}^U |\mathcal{B}_u^t| r_u^t \right)}{N^2 \sum_{u=1}^U |\mathcal{B}_u^t| r_u^t} \geq 0$ and $\langle \sum_{u \in \mathcal{R}_1^t} \mathfrak{R} \left( \nabla J_u \left( \boldsymbol{g}_t \right) \right), \sum_{u \in \mathcal{R}_2^t} \mathfrak{R} \left( \nabla J_u \left( \boldsymbol{g}_t \right) \right) \rangle \leq \left| \langle \sum_{u \in \mathcal{R}_1^t} \mathfrak{R} \left( \nabla J_u \left( \boldsymbol{g}_t \right) \right), \sum_{u \in \mathcal{R}_2^t} \mathfrak{R} \left( \nabla J_u \left( \boldsymbol{g}_t \right) \right) \rangle \right|$, we have 
\begin{equation}
    \begin{split}\label{prf:re_o_1}
        &\mathbb{E} \left[ \lVert \mathfrak{R} \left( \boldsymbol{o} \right) \rVert^2 \right]\\
        &\leq \mathbb{E} \left[ \frac{ \left( N - \sum_{u=1}^U |\mathcal{B}_u^t| r_u^t\right) \sum_{u \in \mathcal{R}_1^t} \lVert \mathfrak{R} \left( \nabla J_u \left( \boldsymbol{g}_t \right) \right) \rVert^2}{N \sum_{u=1}^U |\mathcal{B}_u^t| r_u^t} \right. \\
        & \quad \left. + \frac{\sum_{u \in \mathcal{R}_2^t} \lVert \mathfrak{R} \left( \nabla J_u \left( \boldsymbol{g}_t \right) \right) \rVert^2}{N} + \frac{2\left( N - \sum_{u=1}^U |\mathcal{B}_u^t| r_u^t \right)}{N^2 \sum_{u=1}^U |\mathcal{B}_u^t| r_u^t} \right. \\
        & \quad \left. \left| \langle \sum_{u \in \mathcal{R}_1^t} \mathfrak{R} \left( \nabla J_u \left( \boldsymbol{g}_t \right) \right), \sum_{u \in \mathcal{R}_2^t} \mathfrak{R} \left( \nabla J_u \left( \boldsymbol{g}_t \right) \right) \rangle \right| \right] \\
        &= \mathbb{E} \left[ \frac{\left( N - \sum_{u=1}^U |\mathcal{B}_u^t| r_u^t\right) \sum_{u \in \mathcal{R}_1^t} \lVert \mathfrak{R} \left( \nabla J_u \left( \boldsymbol{g}_t \right) \right) \rVert}{N \sum_{u=1}^U |\mathcal{B}_u^t| r_u^t} \right. \\
        & \quad \left. + \frac{\sum_{u \in \mathcal{R}_2^t} \lVert \mathfrak{R} \left( \nabla J_u \left( \boldsymbol{g}_t \right) \right) \rVert}{N} \right]^2.
    \end{split}
\end{equation}
 Next, we simplify (\ref{prf:re_o_1}) by simplifying $\sum_{u \in \mathcal{R}_1^t} \lVert \mathfrak{R} \left( \nabla J_u \left( \boldsymbol{g}_t \right) \right) \rVert$ and $\sum_{u \in \mathcal{R}_2^t} \lVert \mathfrak{R} \left( \nabla J_u \left( \boldsymbol{g}_t \right) \right) \rVert$. From (\ref{eq:bound1}), since $\mathfrak{R} \left( \nabla J_u \left( \boldsymbol{g}_t \right) \right) = \sum_{k=1}^{|\mathcal{B}_u^t|} \mathfrak{R} \left( \nabla J_u \left( \boldsymbol{g}_t, \boldsymbol{H}_{u,k}, \boldsymbol{p}_{u,k} \right) \right)$, we have $\lVert \mathfrak{R} \left( \nabla J_u \left( \boldsymbol{g}_t \right) \right) \rVert \leq |\mathcal{B}_u^t| \sqrt{\zeta_1 + \zeta_2 \lVert \nabla J \left( \boldsymbol{g}_t \right) \rVert^2}$. Hence, we have 
\begin{equation}\label{eq:r1}
        \sum_{u \in \mathcal{R}_1^t} \lVert \mathfrak{R} \left( \nabla J_u \left( \boldsymbol{g}_t \right) \right) \rVert \leq \sum_{u=1}^U |\mathcal{B}_u^t| r_u^t \sqrt{\zeta_1 + \zeta_2 \lVert \nabla J \left( \boldsymbol{g}_t \right) \rVert^2}, 
\end{equation}
and
\begin{equation}\label{eq:r2}
    \begin{split}
        \sum_{u \in \mathcal{R}_2^t} \lVert \mathfrak{R} &\left( \nabla J_u \left( \boldsymbol{g}_t \right) \right) \rVert \\
        &\leq \left( N-\sum_{u=1}^U |\mathcal{B}_u^t| r_u^t \right) \sqrt{\zeta_1 + \zeta_2 \lVert \nabla J \left( \boldsymbol{g}_t \right) \rVert^2}. 
    \end{split}
\end{equation}
Substituting (\ref{eq:r1}) and (\ref{eq:r2}) into (\ref{prf:re_o_1}), $\mathbb{E} \left[ \lVert \mathfrak{R} \left( \boldsymbol{o} \right) \rVert^2 \right]$ can be expressed by 
\begin{equation}
    \begin{split}
        \mathbb{E} &\left[ \lVert \mathfrak{R} \left( \boldsymbol{o} \right) \rVert^2 \right] \\
        & \leq \frac{4}{N^2} \mathbb{E} \left( N-\sum_{u=1}^U |\mathcal{B}_u^t| r_u^t \right)^2 \left( \zeta_1 + \zeta_2 \lVert \nabla J \left( \boldsymbol{g}_t \right) \rVert^2 \right). 
    \end{split}
\end{equation}
Since $N \geq N-\sum_{u=1}^U |\mathcal{B}_u^t| r_u^t \geq 0$, we have 
\begin{equation}\label{prf:re_o}
    \begin{split}
        \mathbb{E} &\left[ \lVert \mathfrak{R} \left( \boldsymbol{o} \right) \rVert^2 \right] \\
        & \leq \frac{4}{N} \mathbb{E} \left( N-\sum_{u=1}^U |\mathcal{B}_u^t| r_u^t \right) \left( \zeta_1 + \zeta_2 \lVert \nabla J \left( \boldsymbol{g}_t \right) \rVert^2 \right). 
    \end{split}
\end{equation}
Similarly, $\mathbb{E} \left[ \lVert \mathfrak{I} \left( \boldsymbol{o} \right) \rVert^2 \right]$ can be calculated using the same method that used to calculate $\mathbb{E} \left[ \lVert \mathfrak{R} \left( \boldsymbol{o} \right) \rVert^2 \right]$, as follows:  
\begin{equation}\label{prf:im_o}
    \begin{split}
        \mathbb{E} &\left[ \lVert \mathfrak{I} \left( \boldsymbol{o} \right) \rVert^2 \right] \\
        & \leq \frac{4}{N} \mathbb{E} \left( N-\sum_{u=1}^U |\mathcal{B}_u^t| m_u^t \right) \left( \zeta_1 + \zeta_2 \lVert \nabla J \left( \boldsymbol{g}_t \right) \rVert^2 \right). 
    \end{split}
\end{equation}
Based on (\ref{prf:re_o}) and (\ref{prf:im_o}), $\mathbb{E} \left[ \lVert \boldsymbol{o} \rVert^2 \right]$ can be expressed by     
\begin{equation}\label{eq:o}
    \begin{split}
        \mathbb{E} \left[ \lVert \boldsymbol{o} \rVert^2 \right] \leq &\frac{4}{N} \left[ 2N - \mathbb{E} \left( \sum_{u=1}^U |\mathcal{B}_u^t| r_u^t \right) \right. \\
        & \left. + \mathbb{E} \left( \sum_{u=1}^U |\mathcal{B}_u^t| m_u^t \right) \right] \left( \zeta_1 + \zeta_2 \lVert \nabla J \left( \boldsymbol{g}_t \right) \rVert^2 \right). 
    \end{split}
\end{equation}
Substituting (\ref{eq:o}) into (\ref{prf:j_1}), we have
\begin{equation}\label{prf:j_2}
    \begin{split}
        \mathbb{E} &\left( J(\boldsymbol{g}_{t+1}) \right) \\
        & \leq \mathbb{E} \left( J(\boldsymbol{g}_{t}) \right) + \frac{2 \zeta_1 E}{ZN} - \frac{1}{2Z} \left( 1- \frac{4 \zeta_2 E}{N} \right) \lVert \nabla J \left( \boldsymbol{g}_t \right) \rVert^2, 
    \end{split}
\end{equation}
where $E = 2N - \mathbb{E} \left( \sum_{u=1}^U |\mathcal{B}_u^t| r_u^t + \sum_{u=1}^U |\mathcal{B}_u^t| m_u^t \right)$. To show that $J(\boldsymbol{g}_{t}) - J(\overline{\boldsymbol{g}})$ has an upper bound, we subtract $\mathbb{E} \left( J(\overline{\boldsymbol{g}}) \right)$ in both sides of (\ref{prf:j_2}), as follows:  
\begin{equation}\label{prf:j-j*}
    \begin{split}
        \mathbb{E} \left( J(\boldsymbol{g}_{t+1}) - J(\overline{\boldsymbol{g}}) \right) \leq &\mathbb{E} \left( J(\boldsymbol{g}_{t}) - J(\overline{\boldsymbol{g}}) \right) + \frac{2 \zeta_1 E}{ZN} \\
        & - \frac{1}{2Z} \left( 1- \frac{4 \zeta_2 E}{N} \right) \lVert \nabla J \left( \boldsymbol{g}_t \right) \rVert^2.  
    \end{split}
\end{equation}
Then, we simplify (\ref{prf:j-j*}) by simplifing $\lVert \nabla J \left( \boldsymbol{g}_t \right) \rVert^2$. From (\ref{eq:convex}), we have \cite{boyd2004convex}
\begin{equation}\label{eq:convex_1}
    \lVert \nabla J \left( \boldsymbol{g}_t \right) \rVert^2 \geq 2 \mu \left( J(\boldsymbol{g}_{t}) - J(\overline{\boldsymbol{g}}) \right). 
\end{equation}
Substituting (\ref{eq:convex_1}) into (\ref{prf:j-j*}), we have 
\begin{equation}\label{prf:j-j*_1}
    \mathbb{E} \left( J(\boldsymbol{g}_{t+1}) - J(\overline{\boldsymbol{g}}) \right) \leq A \mathbb{E} \left( J(\boldsymbol{g}_{t}) - J(\overline{\boldsymbol{g}}) \right) + \frac{2 \zeta_1 E}{ZN}, 
\end{equation}
where $A = 1-\frac{\mu}{Z}+\frac{4\mu\zeta_2E}{NZ}$. Apply (\ref{prf:j-j*_1}) recursively, we have 
\begin{equation}
    \begin{split}
        \mathbb{E} &\left( J(\boldsymbol{g}_{t+1}) - J(\overline{\boldsymbol{g}}) \right) \\
        & \leq A^{t}  \mathbb{E} \left( J(\boldsymbol{g}_{1}) - J(\overline{\boldsymbol{g}}) \right) + \sum_{a=1}^{t} A^{a} \frac{2 \zeta_1 E}{ZN} \\
        & = A^{t}  \mathbb{E} \left( J(\boldsymbol{g}_{1}) - J(\overline{\boldsymbol{g}}) \right) + \left( \frac{1-A^{t}}{1-A} \right) \frac{2 \zeta_1 E}{ZN}. 
    \end{split}
\end{equation}
This completes the proof.

\bibliographystyle{IEEEbib}
\bibliography{references1}

\begin{thebibliography}{10}

\bibitem{7762095}
A.~Yassin, Y.~Nasser, M.~Awad, A.~Al-Dubai, R.~Liu, C.~Yuen, R.~Raulefs, and
  E.~Aboutanios,
\newblock ``Recent advances in indoor localization: {A} survey on theoretical
  approaches and applications,''
\newblock {\em IEEE Communications Surveys \& Tutorials}, vol. 19, no. 2, pp.
  1327--1346, Secondquarter 2017.

\bibitem{8451859}
B.~Jang and H.~Kim,
\newblock ``Indoor positioning technologies without offline fingerprinting map:
  {A} survey,''
\newblock {\em IEEE Communications Surveys \& Tutorials}, vol. 21, no. 1, pp.
  508--525, Firstquarter 2019.

\bibitem{8057286}
H.~Zou, M.~Jin, H.~Jiang, L.~Xie, and C.~J. Spanos,
\newblock ``Win{IPS}: {W}i{F}i-based non-intrusive indoor positioning system
  with online radio map construction and adaptation,''
\newblock {\em IEEE Transactions on Wireless Communications}, vol. 16, no. 12,
  pp. 8118--8130, December 2017.

\bibitem{zhu2024survey}
Zhiyu Zhu, Yang Yang, Mingzhe Chen, Caili Guo, Julian Cheng, and Shuguang Cui,
\newblock ``A survey on indoor visible light positioning systems: Fundamentals,
  applications, and challenges,''
\newblock {\em arXiv preprint arXiv:2401.13893}, 2024.

\bibitem{8692423}
F.~Zafari, A.~Gkelias, and K.~K. Leung,
\newblock ``A survey of indoor localization systems and technologies,''
\newblock {\em IEEE Communications Surveys \& Tutorials}, vol. 21, no. 3, pp.
  2568--2599, Thirdquarter 2019.

\bibitem{9149443}
A.~Sobehy, E.~Renault, and P.~Muhlethaler,
\newblock ``{CSI-MIMO}: {K}-nearest neighbor applied to indoor localization,''
\newblock in {\em Proc. IEEE International Conference on Communications (ICC)},
  Dublin, Ireland, June 2020.

\bibitem{7438932}
X.~Wang, L.~Gao, S.~Mao, and S.~Pandey,
\newblock ``{CSI}-based fingerprinting for indoor localization: A deep learning
  approach,''
\newblock {\em IEEE Transactions on Vehicular Technology}, vol. 66, no. 1, pp.
  763--776, January 2017.

\bibitem{9129126}
S.~Bast, A.~P. Guevara, and S.~Pollin,
\newblock ``{CSI}-based positioning in massive {MIMO} systems using
  convolutional neural networks,''
\newblock in {\em Proc. IEEE Vehicular Technology Conference}, Antwerp,
  Belgium, May 2020.

\bibitem{8027020}
H.~Chen, Y.~Zhang, W.~Li, X.~Tao, and P.~Zhang,
\newblock ``Con{F}i: Convolutional neural networks based indoor {W}i-{F}i
  localization using channel state information,''
\newblock {\em IEEE Access}, vol. 5, pp. 18066--18074, September 2017.

\bibitem{9535306}
E.~G\"{o}n\"{u}lta\c{s}, E.~Lei, J.~Langerman, H.~Huang, and C.~Studer,
\newblock ``{CSI}-based multi-antenna and multi-point indoor positioning using
  probability fusion,''
\newblock {\em IEEE Transactions on Wireless Communications}, vol. 21, no. 4,
  pp. 2162--2176, April 2022.

\bibitem{8919897}
E.~Lei, O.~Casta\~{n}eda, O.~Tirkkonen, T.~Goldstein, and C.~Studer,
\newblock ``Siamese neural networks for wireless positioning and channel
  charting,''
\newblock in {\em Proc. Annual Allerton Conference on Communication, Control,
  and Computing (Allerton)}, Monticello, IL, USA, September 2019, pp. 200--207.

\bibitem{9999279}
Y.~Ruan, L.~Chen, X.~Zhou, Z.~Liu, X.~Liu, G.~Guo, and R.~Chen,
\newblock ``i{P}os-5{G}: {I}ndoor positioning via commercial 5{G} {NR} {CSI},''
\newblock {\em IEEE Internet of Things Journal}, vol. 10, no. 10, pp.
  8718--8733, May 2023.

\bibitem{8761305}
H.~Zhang, H.~Du, Q.~Ye, and C.~Liu,
\newblock ``Utilizing {CSI} and {RSSI} to achieve high-precision outdoor
  positioning: {A} deep learning approach,''
\newblock in {\em Proc. IEEE International Conference on Communications (ICC)},
  Shanghai, China, May 2019, pp. 1--6.

\bibitem{9066152}
Y.~Liu, H.~Li, J.~Xiao, and H.~Jin,
\newblock ``F{L}oc: {F}ingerprint-based indoor localization system under a
  federated learning updating framework,''
\newblock in {\em International Conference on Mobile Ad-Hoc and Sensor Networks
  (MSN)}, Shenzhen, China, December 2019, pp. 113--118.

\bibitem{9917443}
B.~Gao, F.~Yang, N.~Cui, K.~Xiong, Y.~Lu, and Y.~Wang,
\newblock ``A federated learning framework for fingerprinting-based indoor
  localization in multibuilding and multifloor environments,''
\newblock {\em IEEE Internet of Things Journal}, vol. 10, no. 3, pp.
  2615--2629, February 2023.

\bibitem{10214616}
F.~Dou, J.~Lu, T.~Zhu, and J.~Bi,
\newblock ``On-device indoor positioning: {A} federated reinforcement learning
  approach with heterogeneous devices,''
\newblock {\em IEEE Internet of Things Journal}, vol. 11, no. 3, pp.
  3909--3926, February 2024.

\bibitem{9593115}
P.~Wu, T.~Imbiriba, J.~Park, S.~Kim, and P.~Closas,
\newblock ``Personalized federated learning over non-{IID} data for indoor
  localization,''
\newblock in {\em Proc. IEEE International Workshop on Signal Processing
  Advances in Wireless Communications (SPAWC)}, Lucca, Italy, September 2021,
  pp. 421--425.

\bibitem{10118848}
Y.~Etiabi, W.~Njima, and E.~M. Amhoud,
\newblock ``Federated learning based hierarchical 3{D} indoor localization,''
\newblock in {\em Proc. IEEE Wireless Communications and Networking Conference
  (WCNC)}, Glasgow, United Kingdom, March 2023, pp. 1--6.

\bibitem{9838945}
N.~Nagia, M.~T. Rahman, and S.~Valaee,
\newblock ``Federated learning for {W}i{F}i fingerprinting,''
\newblock in {\em Proc. IEEE International Conference on Communications},
  Seoul, Korea, Republic of, May 2022, pp. 4968--4973.

\bibitem{10005038}
J.~Guo, I.~W.~H. Ho, Y.~Hou, and Z.~Li,
\newblock ``Fed{P}os: {A} federated transfer learning framework for {CSI}-based
  {W}i-{F}i indoor positioning,''
\newblock {\em IEEE Systems Journal}, vol. 17, no. 3, pp. 4579--4590, September
  2023.

\bibitem{9148111}
B.~S. Ciftler, A.~Albaseer, N.~Lasla, and M.~Abdallah,
\newblock ``Federated learning for {RSS} fingerprint-based localization: {A}
  privacy-preserving crowdsourcing method,''
\newblock in {\em Proc. International Wireless Communications and Mobile
  Computing (IWCMC)}, Limassol, Cyprus, June 2020, pp. 2112--2117.

\bibitem{9562559}
M.~Chen, D.~Gunduz, K.~Huang, W.~Saad, M.~Bennis, A.~V. Feljan, and H.~V. Poor,
\newblock ``Distributed learning in wireless networks: {R}ecent progress and
  future challenges,''
\newblock {\em IEEE Journal on Selected Areas in Communications}, vol. 39, no.
  12, pp. 3579--3605, December 2021.

\bibitem{9413814}
J.~A. Barrachina, C.~Ren, C.~Morisseau, G.~Vieillard, and J.-P. Ovarlez,
\newblock ``Complex-valued vs. real-valued neural networks for classification
  perspectives: {A}n example on non-circular data,''
\newblock in {\em Proc. IEEE International Conference on Acoustics, Speech and
  Signal Processing (ICASSP)}, Toronto, ON, Canada, June 2021, pp. 2990--2994.

\bibitem{benvenuto1992complex}
N.~Benvenuto and F.~Piazza,
\newblock ``On the complex backpropagation algorithm,''
\newblock {\em IEEE Transactions on Signal Processing}, vol. 40, no. 4, pp.
  967--969, April 1992.

\bibitem{9076084}
H.~Ye, F.~Gao, J.~Qian, H.~Wang, and G.~Li,
\newblock ``Deep learning-based denoise network for {CSI} feedback in {FDD}
  massive {MIMO} systems,''
\newblock {\em IEEE Communications Letters}, vol. 24, no. 8, pp. 1742--1746,
  April 2020.

\bibitem{1703954}
Y.~Qi, H.~Kobayashi, and H.~Suda,
\newblock ``On time-of-arrival positioning in a multipath environment,''
\newblock {\em IEEE Transactions on Vehicular Technology}, vol. 55, no. 5, pp.
  1516--1526, September 2006.

\bibitem{8264077}
S.~Khirirat, H.~R. Feyzmahdavian, and M.~Johansson,
\newblock ``Mini-batch gradient descent: {F}aster convergence under data
  sparsity,''
\newblock in {\em Proc. IEEE Annual Conference on Decision and Control (CDC)},
  Melbourne, VIC, Australia, December 2017.

\bibitem{chen2021communication}
M.~Chen, N.~Shlezinger, H.~V. Poor, Y.~C. Eldar, and S.~Cui,
\newblock ``Communication-efficient federated learning,''
\newblock {\em Proceedings of the National Academy of Sciences}, vol. 118, no.
  17, pp. e2024789118, April 2021.

\bibitem{mcmahan2017communication}
B.~McMahan, E.~Moore, D.~Ramage, S.~Hampson, and B.~A. y~Arcas,
\newblock ``Communication-efficient learning of deep networks from
  decentralized data,''
\newblock in {\em Proc. Artificial Intelligence and Statistics}, Fort
  Lauderdale, FL, USA, April 2017, pp. 1273--1282.

\bibitem{chen2020joint}
M.~Chen, Z.~Yang, W.~Saad, C.~Yin, H.~V. Poor, and S.~Cui,
\newblock ``A joint learning and communications framework for federated
  learning over wireless networks,''
\newblock {\em IEEE Transactions on Wireless Communications}, vol. 20, no. 1,
  pp. 269--283, October 2020.

\bibitem{amiri2021convergence}
M.~M. Amiri, D~G{\"u}nd{\"u}z, S.~R. Kulkarni, and H.~V. Poor,
\newblock ``Convergence of update aware device scheduling for federated
  learning at the wireless edge,''
\newblock {\em IEEE Transactions on Wireless Communications}, vol. 20, no. 6,
  pp. 3643--3658, June 2021.

\bibitem{qiu2016going}
J.~Qiu, J.~Wang, S.~Yao, K.~Guo, B.~Li, E.~Zhou, J.~Yu, T.~Tang, N.~Xu,
  S.~Song, et~al.,
\newblock ``Going deeper with embedded {FPGA} platform for convolutional neural
  network,''
\newblock in {\em Proc. ACM/SIGDA International Symposium on Field-programmable
  Gate Arrays}, Monterey, CA, USA, February 2016, pp. 26--35.

\bibitem{pan2023situ}
M.~Pan, S.~Liu, P.~Liu, W.~Qi, Y.~Huang, W.~Zheng, Q.~Wu, and M.~Gardill,
\newblock ``In situ calibration of antenna arrays for positioning with 5{G}
  networks,''
\newblock {\em IEEE Transactions on Microwave Theory and Techniques}, vol. 71,
  no. 10, pp. 4600--4613, October 2023.

\bibitem{9535488}
C.~Li, S.~De~Bast, E.~Tanghe, S.~Pollin, and W.~Joseph,
\newblock ``Toward fine-grained indoor localization based on massive
  {MIMO-OFDM} system: Experiment and analysis,''
\newblock {\em IEEE Sensors Journal}, vol. 22, no. 6, pp. 5318--5328, March
  2022.

\bibitem{8297024}
P.~Virtue, S.~X. Yu, and M.~Lustig,
\newblock ``Better than real: {C}omplex-valued neural nets for {MRI}
  fingerprinting,''
\newblock in {\em Proc. IEEE International Conference on Image Processing
  (ICIP)}, Beijing, China, September 2017, pp. 3953--3957.

\bibitem{karnewar2022relu}
A.~Karnewar, T.~Ritschel, O.~Wang, and N.~Mitra,
\newblock ``Relu fields: {T}he little non-linearity that could,''
\newblock in {\em Proc. ACM Computer Graphics and Interactive Techniques
  Conference}, Vancouver, BC, Canada, August 2022, pp. 1--9.

\bibitem{boyd2004convex}
S.~Boyd and L.~Vandenberghe,
\newblock {\em \emph{Convex Optimization}},
\newblock Cambridge University Press, 2004.

\end{thebibliography}
\end{document}